\documentclass[3p,11pt,fleqn,sort&compress]{elsarticle}

\usepackage{latexsym}
\usepackage{xspace}
\usepackage{amssymb}
\usepackage{amsmath}
\usepackage{stmaryrd}
\usepackage{hyperref}
\usepackage{url}
\usepackage{graphicx}
\input{xy}
\xyoption{all}
\usepackage{calc}
\usepackage{tikz}
\usepackage{flushend}
\usepackage{color}
\usepackage[ruled,vlined,linesnumbered]{algorithm2e}
\usepackage{csquotes}
\usepackage{tabularx}
\usepackage{nicematrix}
\usepackage{booktabs}

\newtheorem{theorem}{Theorem}[section]
\newtheorem{lemma}[theorem]{Lemma}
\newtheorem{proposition}[theorem]{Proposition}
\newtheorem{corollary}[theorem]{Corollary}

\newtheorem{Definition}[theorem]{Definition}
\newtheorem{Example}[theorem]{Example}
\newtheorem{Remark}[theorem]{Remark}

\newenvironment{example}{\begin{Example}\begin{em}}{\end{em}\end{Example}}

\newproof{proof}{Proof}

\usepackage{tikz,tikz-3dplot,pgfplots}
\pgfplotsset{compat=newest}
\usetikzlibrary{shapes.geometric,arrows,arrows.meta}
\usetikzlibrary{fadings}
\tikzfading %strangely gives bad bounding box when inside the tikzpicture
[
  name=fade out,
  inner color=transparent!0,
  outer color=transparent!100
]
\usetikzlibrary{arrows.meta,automata,positioning,shadows}

\tikzset{
    invisible/.style={opacity=0},
    visible on/.style={alt={#1{}{invisible}}},
    alt/.code args={<#1>#2#3}{%
      \alt<#1>{\pgfkeysalso{#2}}{\pgfkeysalso{#3}} % \pgfkeysalso doesn't change the path
    },
  }
\usetikzlibrary{calc,shapes}

\def\eqref#1{(\ref{#1})}

\def\tuple#1{\langle#1\rangle}

\newcommand{\mL}{\mathcal{L}}
\newcommand{\bL}{\mathbf{L}}
\newcommand{\bLdb}{\bL^{\leq n}}

\newcommand{\mA}{\mathcal{A}}
\newcommand{\mAp}{{\mathcal{A}'}}
\newcommand{\mAdp}{{\mathcal{A}''}}

\newcommand{\mB}{\mathcal{B}}

\newcommand{\NN}{\mathbb{N}}

\newcommand{\myend}{\mbox{}\hfill{\small$\blacksquare$}}

\newcommand{\comment}[1]{}

\newcommand{\fand}{\varotimes}

\newcommand{\fto}{\Rightarrow}
\newcommand{\fequiv}{\Leftrightarrow}

\newcommand{\deltaA}{\delta^\mA}
\newcommand{\sigmaA}{\sigma^\mA}
\newcommand{\tauA}{\tau^\mA}
\newcommand{\deltaAp}{\delta^\mAp}
\newcommand{\sigmaAp}{\sigma^\mAp}
\newcommand{\tauAp}{\tau^\mAp}
\newcommand{\deltaAdp}{\delta^\mAdp}
\newcommand{\sigmaAdp}{\sigma^\mAdp}
\newcommand{\tauAdp}{\tau^\mAdp}

\newcommand{\normS}[3]{\|#1\|_{{#2} \fto {#3}}}
\newcommand{\normBS}[3]{\|#1\|_{{#2} \fequiv {#3}}}

\newcommand{\nZs}{\normS{\varphi}{\mA}{\mAp}}
\newcommand{\nZbs}{\normBS{\varphi}{\mA}{\mAp}}

\newcommand{\nZdbs}{\normS{\Phi}{\mA}{\mAp}}
\newcommand{\nZdbbs}{\normBS{\Phi}{\mA}{\mAp}}

\newcommand{\mF}{\mathcal{F}}
\newcommand{\mFs}{\mF_{\!\!_\to}(\Sigma,\mL)}
\newcommand{\mFbs}{\mF_{\!\!_\leftrightarrow}(\Sigma,\mL)}

\newcommand{\mFdbs}[1]{\mF^{\,\leq #1}_{\!\!_\to}(\Sigma,\mL)}
\newcommand{\mFdbbs}[1]{\mF^{\,\leq #1}_{\!\!_\leftrightarrow}(\Sigma,\mL)}

\newcommand{\Succ}{\mathit{succ}}
\newcommand{\Pred}{\mathit{pred}}

\newcommand{\CompDBFS}{\mbox{$\mathsf{ComputeDepthBoundedFuzzySimulation}$}\xspace}
\newcommand{\CompDBFB}{\mbox{$\mathsf{ComputeDepthBoundedFuzzyBisimulation}$}\xspace}

\SetKwInput{Input}{Input}
\SetKwInput{Output}{Output}
\SetKwInput{LocalVariables}{Local variables}
\SetKwInput{Purpose}{Purpose}
\SetKw{Break}{break}
\SetKw{Continue}{continue}
\SetKw{Return}{return}

\newcommand{\True}{\mathit{true}}
\newcommand{\False}{\mathit{false}}
\newcommand{\changed}{\mathit{changed}}

\journal{arXiv}

\begin{document}
\sloppy
	
\begin{frontmatter}		
	
\title{Depth-Bounded Fuzzy Simulations and Bisimulations between~Fuzzy~Automata}
%\title{Depth-Bounded Fuzzy Simulations and Bisimulations}

\author[inst1,inst2]{Linh Anh Nguyen}
\ead{nguyen@mimuw.edu.pl}
\author[inst3]{Ivana Mici\' c}
\ead{ivana.micic@pmf.edu.rs}
%\author[inst3]{Stefan Stanimirovi\' c\corref{labcor}}
\author[inst3]{Stefan Stanimirovi\' c}
\ead{stefan.stanimirovic@pmf.edu.rs}
%\cortext[labcor]{Corresponding author.}

\address[inst1]{Institute of Informatics, University of Warsaw, Banacha 2, 02-097 Warsaw, Poland}
\address[inst2]{Faculty of Information Technology, Nguyen Tat Thanh University, Ho Chi Minh City, Vietnam}
\address[inst3]{University of Ni\v s, Faculty of Sciences and Mathematics, Vi\v segradska 33, 18000 Ni\v s, Serbia}

\begin{abstract}
Simulations and bisimulations are well-established notions in crisp/fuzzy automata theory and are widely used to compare the behaviors of automata. Their main drawback is that they compare the behaviors of fuzzy automata in a crisp manner. Recently, fuzzy simulations and fuzzy bisimulations have been defined for fuzzy automata as a kind of approximate simulations and approximate bisimulations that compare the behaviors of fuzzy automata in a fuzzy manner. However, they still suffer from serious shortcomings. First, they still cannot correlate all fuzzy automata that are intuitively \enquote{more or less} (bi)similar. Second, the currently known algorithms for computing the greatest fuzzy simulation or bisimulation between two finite fuzzy automata have an exponential time complexity when the {\L}ukasiewicz or product structure of fuzzy values is used. This work deals with these problems, providing approximations of fuzzy simulations and fuzzy bisimulations. We define such approximations via a novel notion of decreasing sequences of fuzzy relations whose infima are, under some conditions, fuzzy simulations (respectively, bisimulations). We call such a sequence a depth-bounded fuzzy simulation (respectively, bisimulation), as the $n$th element from the sequence compares the behaviors of fuzzy automata, but only for words with a length bounded by~$n$. We further provide a logical characterization of the greatest depth-bounded fuzzy simulation or bisimulation between two fuzzy automata by proving that it satisfies the corresponding Hennessy-Milner property. Finally, we provide polynomial-time algorithms for computing the $n$th component of the greatest depth-bounded fuzzy simulation (respectively, bisimulation) between two finite fuzzy automata.
\end{abstract}


\end{frontmatter}

%===============================================================================

\section{Introduction}
\label{section:intro}

Bisimulation is a concept in computer science and mathematical logic used to define the notion of equivalence between systems. 
It is particularly useful in the study of automata, which are mathematical models that represent systems that change states under the influence of the input symbols. Fuzzy automata are a type of automata that allow for imprecise or uncertain transitions, initial and final states, making them useful in applications where precise transitions are difficult to obtain, such as decision-making systems, medical systems and natural language processing. 
In the context of fuzzy automata, bisimulations are used to compare two automata and determine whether they are equivalent. They have been studied both as crisp and fuzzy relations.

Approximate bisimulations are a generalization of bisimulations that allow for a certain degree of error or approximation instead of requiring an exact equivalence between states. In other words, when an approximate bisimulation relates two states, we can treat them as approximately equivalent if they are sufficiently similar, even if they are not identical.
The notion of approximate bisimulations is particularly useful in the context of fuzzy automata, where the states and transitions may be imprecise or uncertain. For example, in a fuzzy automaton that models a temperature control system, the states may be fuzzy sets that represent a range of possible temperature values. In this case, an exact equivalence between states may not be possible or desirable, as slight variations in temperature may not significantly impact the system's behavior.

Approximate bisimulations for fuzzy automata have been first defined as crisp relations~\cite{YL.18b, YL.18a, DBLP:journals/fss/YangL20}. Later, they were defined in~\cite{SMC.20,MNS.22} as fuzzy relations. However, the latter are defined only for fuzzy automata over complete Heyting algebras, and like \cite{YL.18b,YL.18a,DBLP:journals/fss/YangL20}, the greatest $\lambda$-approximate bisimulation between two fuzzy automata defined in~\cite{SMC.20,MNS.22} may not exist. These facts motivated the author in~\cite{FB4FA} to seek a more general definition of approximate bisimulation. More precisely, the work~\cite{FB4FA} introduces fuzzy simulations and fuzzy bisimulations between fuzzy automata over any complete residuated lattice. 
Moreover, the same work demonstrates that fuzzy simulations fuzzily preserve the fuzzy language recognized by a fuzzy automaton, while fuzzy bisimulations keep this fuzzy language fuzzily invariant. In addition, fuzzy simulations and fuzzy bisimulations have the Hennessy-Milner properties. Finally, the greatest fuzzy simulation and bisimulation always exist between two fuzzy automata.

However, fuzzy simulations and bisimulations still suffer from a serious shortcoming. Namely, observe the fuzzy automata $\mA$ and $\mAp$ depicted in Figure~\ref{fig: FSB}. Each of them has only one state ($u$ or~$u'$), which is both initial and final in the degree $1$, and only one transition from that state to itself under the influence of the unique symbol~$s$ from the alphabet. The degree of the transition in $\mA$ is equal to~$1$, whereas in $\mAp$ is equal to $1 - \varepsilon$, where $0 < \varepsilon < 1$ is a parameter. When $\varepsilon$ is very small (i.e., close to 0), we can treat these fuzzy automata as \enquote{almost equal}. Unfortunately, when the product or {\L}ukasiewicz t-norm is used, the greatest fuzzy (bi)simulation between~$\mA$ and~$\mAp$ is equal to the empty fuzzy relation. This unexpected phenomenon motivates us to seek a more refined notion than the one of a fuzzy (bi)simulation.

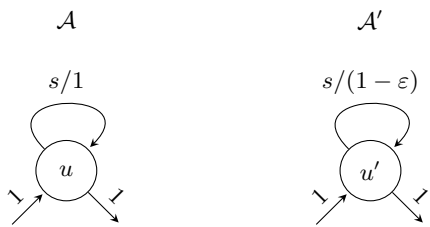
\begin{figure}[!hbt]
\begin{center}
\begin{tikzpicture}[->,>=stealth,shorten >=1pt,node distance=1.2cm,auto]
    \tikzset{every state/.style={inner sep=0.12cm,minimum size=0.8cm}}
    \tikzstyle{every node}=[font=\footnotesize]
    \node (A) {$\mA$};
    \node[state] (q) [below of=A, node distance = 2cm] {$u$};
    \path[->] (q) edge[loop above,in=45,out=135,looseness=5] node{$s/1$} (q);
    \node[below left of=q] (in) {};
    \node[below right of=q] (out) {};
    \path[->](in) edge[above, sloped] node{1} (q);
    \path[->](q) edge[above, sloped] node{1} (out);
    \node (Ap) [node distance=4cm, right of=A] {$\mAp$};
    \node[state] (qp) [below of=Ap, node distance = 2cm] {$u'$};
    \path[->] (qp) edge[loop above,in=45,out=135,looseness=5] node{$s/(1 - \varepsilon)$} (qp);
    \node[below left of=qp] (inp) {};
    \node[below right of=qp] (outp) {};
    \path[->](inp) edge[above, sloped] node{1} (qp);
    \path[->](qp) edge[above, sloped] node{1} (outp);
\end{tikzpicture}
\caption{An illustration of the fuzzy automata $\mA$ and $\mAp$ discussed in Section~\ref{section:intro}.\label{fig: FSB}}
\end{center}
\end{figure}

This work provides a generalization of fuzzy simulations and bisimulations. The proposed generalization is new and original, as it perceives a simulation (respectively, bisimulation) as a {\em decreasing sequence} of fuzzy relations satisfying some conditions. Under certain light assumptions, which we determine in this work, the infimum of such a sequence is a fuzzy simulation (respectively, bisimulation) from~\cite{FB4FA}. Thus, we can treat such a sequence as an \emph{approximation} of a fuzzy simulation (respectively, bisimulation). Moreover, the proposed approximations can measure the extent to which the fuzzy language recognized by a fuzzy automaton is fuzzily preserved by the simulation (respectively, fuzzily invariant under the bisimulation), when restricted to the words bounded by a specific length. Namely, the so-called norm of the $n$th fuzzy relation from the sequence is the infimum of the degrees in which the words with a length bounded by $n$ are fuzzily preserved by the simulation (respectively, fuzzily invariant under the bisimulation). This is why such a sequence is called a {\em depth-bounded} fuzzy simulation (respectively, bisimulation). 

According to our definitions given in this work, the greatest depth-bounded fuzzy (bi)simulation between the fuzzy automata $\mA$ and $\mAp$ discussed above and depicted in Figure~\ref{fig: FSB} is the sequence $(\varphi_n)_{n \in \NN}$ of fuzzy relations (between $\{u\}$ and $\{u'\}$) specified below for the most well-known structures of fuzzy values:
\begin{itemize}
\item when the G\"odel structure is used: $\varphi_0(u,u') = 1$ and $\varphi_n(u,u') = 1-\varepsilon$ for all $n \geq 1$;

\item when the {\L}ukasiewicz structure is used: $\varphi_n(u,u')$ is $1 - n\varepsilon$ if $n \leq 1/\varepsilon$, and 0 otherwise;

\item when the product structure is used: $\varphi_n(u,u') = (1 - \varepsilon)^n$.
\end{itemize}
When the {\L}ukasiewicz or product structure is used, the sequence $(\varphi_n)_{n \in \NN}$ converges to the empty fuzzy relation, implying that if $\varphi$ is the greatest fuzzy (bi)simulation between $\mA$ and $\mAp$, then $\varphi(u,u') = 0$. As stated before, this latter fact is an unexpected phenomenon, as $\varepsilon$ may be very small and just a noise, in which case we should treat the states $u$ and $u'$ as almost (bi)similar and expect $\varphi(u,u')$ to be close to~1. 
If we use $\varphi_n$ instead of $\varphi$, for a not too big~$n$, then $\varphi_n(u,u')$ is indeed very close to~1 when~$\varepsilon$ is very small. Furthermore, for every word $w$ (over the considered alphabet) of a length bounded by $n$, 
\[ \varphi_n(u,u') \leq (\bL(\mA)(w) \fequiv \bL(\mAp)(w)), \]
where $\bL(\mA)(w)$ and $\bL(\mAp)(w)$ are the degrees in which $\mA$ and $\mAp$ accept $w$, respectively, and~$\fequiv$ denotes the involved biresiduum. These facts motivate our intention to introduce depth-bounded fuzzy (bi)simulations as a generalization of fuzzy (bi)simulations. 

Our results are as follows: we relate depth-bounded fuzzy (bi)simulations to fuzzy (bi)simulations and prove the fuzzy preservation (respectively, invariance) of length-bounded fuzzy languages under depth-bounded fuzzy simulations (respectively, bisimulations). We also formulate and prove the Hennessy-Milner properties of depth-bounded fuzzy (bi)simulations, which are logical characterizations of the greatest depth-bounded fuzzy (bi)simulations.

Moreover, we provide algorithms that, given finite fuzzy automata $\mB$ and $\mB'$ together with a natural number $k$, compute the component $\varphi_k$ of the greatest depth-bounded fuzzy simulation (respectively, bisimulation) $(\varphi_i)_{i \in \NN}$ between $\mB$ and $\mB'$. These algorithms have a time complexity of order $O(k(m+n)n)$, where $n$ is the number of states and $m$ is the number of (nonzero) transitions in the input fuzzy automata. As this order is polynomial in $k$, $n$ and $m$, the algorithms are of particular importance in the context that the currently known algorithms (with the termination property) for computing the greatest fuzzy simulation or bisimulation between two finite fuzzy automata have an exponential time complexity when the {\L}ukasiewicz or product structure of fuzzy values is used~\cite{FuzzyMinimaxNets}. 

The rest of this work is structured as follows. 
Section~\ref{section: prel} contains preliminaries. 
In Section~\ref{section: fs-4-fa}, we define and study depth-bounded fuzzy simulations between fuzzy automata. Theorem~\ref{theorem: HDKAK} relates depth-bounded fuzzy simulations to fuzzy simulations. Theorem~\ref{theorem: JHFLS} states the fuzzy preservation of the fuzzy length-bounded languages recognized by a fuzzy automaton under depth-bounded fuzzy simulations. Theorem~\ref{theorem: HDJHA} states the Hennessy-Milner property of depth-bounded fuzzy simulations between fuzzy automata. 
In Section~\ref{section: fbs}, we define and study depth-bounded fuzzy bisimulations between fuzzy automata, giving Theorems~\ref{theorem: HDKAK 2}, \ref{theorem: HGFKW} and~\ref{theorem: HDJHA 2}, which are counterparts of Theorems~\ref{theorem: HDKAK}, \ref{theorem: JHFLS} and~\ref{theorem: HDJHA}, respectively.
In Section~\ref{section: computation}, we present our above-mentioned algorithms.
We discuss related work in Section~\ref{section: related work} and give concluding remarks in Section~\ref{section: cons}.

%===============================================================================

\section{Preliminaries}
\label{section: prel}

In this section, we recall basic definitions related to residuated lattices, fuzzy sets and relations, fuzzy automata, fuzzy simulations and bisimulations. Its contents come from~\cite{FB4FA,NguyenFSS2021} with some extensions.

\subsection{Residuated Lattices}

A {\em residuated lattice} \cite{Hajek1998,Belohlavek2002} is an algebra $\mL = \tuple{L$, $\leq$, $\fand$, $\fto$, $0$, $1}$ such that%\footnote{See also \url{https://en.wikipedia.org/wiki/Residuated_lattice}} 
\begin{itemize}
\item $\tuple{L, \leq, 0, 1}$ is a lattice with the smallest element 0 and the greatest element 1,
\item $\tuple{L, \fand, 1}$ is a commutative monoid with the unit 1, 
\item for every $x, y, z \in L$, the following adjunction property holds:
\begin{equation}
x \fand y \leq z \ \ \textrm{iff}\ \ x \leq (y \fto z). \label{fop: GDJSK 00} 
\end{equation}
\end{itemize}

%The expression \mbox{$y \fto z$} is called the {\em residual} of $z$ by $y$. 
Given a residuated lattice \mbox{$\mL = \tuple{L, \leq, \fand, \fto, 0, 1}$}, the {\em meet} is denoted by $\lor$ and the {\em join} by $\land$. Moreover, the {\em biresiduum} \mbox{$(x \fto y) \land (y \fto x)$} is denoted by \mbox{$x \fequiv y$}. The {\em supremum} and {\em infimum} of a non-empty set $X\subseteq L$, if they exist, are denoted by $\bigvee\!X$ and $\bigwedge\!X$, respectively. Similarly, for $X = \{x_i \mid i \in I\} \subseteq L$, we write $\bigvee_{i \in I} x_i$ and $\bigwedge_{i \in I} x_i$ to denote $\bigvee\!X$ and $\bigwedge\!X$, respectively, if they exist. We adopt the convention that $\land$ and $\fand$ bind stronger than $\lor$, which in turn binds stronger than~$\fequiv$ and~$\fto$. 

A residuated lattice $\mL = \tuple{L, \leq, \fand, \fto, 0, 1}$ is called a {\em Heyting algebra} if $\fand$ and $\land$ are the same. It is said to be {\em linear} (respectively, {\em complete}) if the lattice \mbox{$\tuple{L, \leq, 0, 1}$} is linear (respectively, complete). The operator $\fand$ is {\em continuous} (with respect to infima) if, for every $x \in L$ and $Y \subseteq L$, 
\begin{equation} \label{eq.ContTNorm}
x \fand {\textstyle\bigwedge} Y = \bigwedge_{y \in Y}\!(x \fand y). 
\end{equation}

%We will need the following lemma. 

\begin{lemma}\label{lemma: JHFJW}
	For a residuated lattice $\mL = \tuple{L, \leq, \fand, \fto, 0, 1}$, the following properties hold for each $x,y,z,x',y' \in L$:
	\begin{eqnarray}
	x \leq x' \textrm{ and } y \leq y' & \!\!\textrm{implies}\!\! & x \fand y \leq x' \fand y' \label{fop: GDJSK 10}\\
	x' \leq x \textrm{ and } y \leq y' & \!\!\textrm{implies}\!\! & (x \fto y) \leq (x' \fto y') \label{fop: GDJSK 20}\\
	x \leq y & \textrm{iff} & (x \fto y) = 1 \label{fop: GDJSK 30}
	\end{eqnarray}
	\begin{eqnarray}
	x \fand 0 & = & 0 \label{fop: GDJSK 40} \\
	% x \fand (y \lor z) & = & x \fand y \,\lor\, x \fand z \label{fop: GDJSK 50} \\
	x \fand (x \fto y) & \leq & y \label{fop: GDJSK 60} \\
	% x \fand (y \fto z) & \leq & (x \fto y) \fto z \label{fop: GDJSK 70} \\
	% x \fand (y \fequiv z) & \leq & (x \fto y) \fto z \label{fop: GDJSK 70a} \\
	% x \fand (y \fequiv z) & \leq & (x \fand z \fto y) \label{fop: GDJSK 80b} \\
	% x \fand (y \fequiv z) & \leq & y \,\fequiv\, x \fand z \label{fop: GDJSK 80} \\
	% x \fto (y \fto z) & = & y \fto (x \fto z) \label{fop: GDJSK 90} \\
	x \fto (y \fto z) & \leq & x \fand y \,\fto\, z \label{fop: GDJSK 100} \\
	x \fand (y \fto z) & \leq & y \,\fto\, x \fand z \label{fop: GDJSK 80a} \\
	% x \fto (y \fequiv z) & \leq & x \fand y \,\fto\, z \label{fop: GDJSK 110} \\
	% (x \fto y) \fand (y \fto z) & \leq & x \fto z \label{fop: GDJSK 115} \\
	% (x \fequiv y) \fand (y \fequiv z) & \leq & x \fequiv z \label{fop: GDJSK 115a} \\
	% (x \fequiv x') \land (y \fequiv y') & \leq & x \land y \,\fequiv\, x' \land y' \label{fop: GDJSK 120} \\
	% (x \fequiv x') \land (y \fequiv y') & \leq & x \lor y \,\fequiv\, x' \lor y' \label{fop: GDJSK 130} \\
	% x \fequiv y & \leq & (z \fto x) \fequiv (z \fto y) \label{fop: GDJSK 140} \\
	% x \fequiv y & \leq & (x \fto z) \fequiv (y \fto z) \label{fop: GDJSK 150} \\
	x \fequiv y & \leq & (z \fequiv x) \fequiv (z \fequiv y). \label{fop: GDJSK 150b}
	\end{eqnarray}
	If $\mL$ is complete, then the following properties hold for all $x,y \in L$ and $X,Y \subseteq L$:
	\begin{eqnarray}
	x \fand \textstyle\bigvee\!Y & = & \bigvee_{y \in Y} (x \fand y) \label{fop: GDJSK 155} \\
	%x \fto \sup Y & = & \sup \{x \fto y \mid y \in Y\} \label{fop: GDJSK 200} \\
	(\textstyle\bigvee\!X) \fto y & = & \bigwedge_{x \in X} (x \fto y) \label{fop: GDJSK 210} \\
    \bigvee_{y \in Y} (x \fto y) & \le &   x \fto \textstyle\bigvee\!Y \label{fop: GDJSK 230} \\
	x \fto \textstyle\bigwedge\!Y & = & \bigwedge_{y \in Y} (x \fto y). \label{fop: GDJSK 220}
	\end{eqnarray}
\end{lemma}

\begin{proof}
The proofs of~\eqref{fop: GDJSK 10}--\eqref{fop: GDJSK 210} can be found in other sources, e.g.~\cite{FB4FA,NguyenFSS2021}. 
The assertion~\eqref{fop: GDJSK 230} directly follows from~\eqref{fop: GDJSK 20}.
We present here a proof of~\eqref{fop: GDJSK 220} although such a proof can be found in the literature. 
The fact that the LHS of~\eqref{fop: GDJSK 220} is less than or equal to the RHS follows directly from~\eqref{fop: GDJSK 20}. To prove the converse, note that $x \fand \bigwedge_{y' \in Y} (x \fto y') \leq y$ follows from~\eqref{fop: GDJSK 10} and~\eqref{fop: GDJSK 60}, for every $x \in L$ and $y \in Y$. Therefore, $x \fand \bigwedge_{y' \in Y} (x \fto y') \leq \bigwedge\!Y$. By~\eqref{fop: GDJSK 00}, it follows that $\bigwedge_{y \in Y} (x \fto y) \leq (x \fto \bigwedge\!Y)$, which completes the proof.
\myend
\end{proof}

\begin{example}\label{exm.CRL}
The most studied residuated lattices are the unit interval $L=[0,1]$ with the usual order and the operator $\fand$ that is one of the t-norms specified below together with its corresponding residuum:
\begin{itemize}
    \item $x \fand y = \min\{x,y\}$ and $(x \fto y) = 
		\left\{
		\!\!\!\begin{array}{ll}
		1 & \!\textrm{if $x \leq y$} \\ 
		y & \!\textrm{otherwise}
		\end{array}\!\!\!
		\right.$ (the G\"odel structure);
    \item $x \fand y = \max\{0, x+y-1\}$ and $(x \fto y) = 
		\min\{1, 1 - x + y\}$ (the {\L}ukasiewicz structure);
    \item $x \fand y = x \cdot y$ and $(x \fto y) = 
		\left\{
		\!\!\!\begin{array}{ll}
		1 & \!\textrm{if $x \leq y$} \\ 
		y/x & \!\textrm{otherwise}
		\end{array}\!\!\!
		\right.$ (the product structure).
\end{itemize}
Note that the above defined structures are linear and complete, with the corresponding t-norms $\fand$ being continuous.
\myend
\end{example}

From now on, let $\mL = \tuple{L, \leq, \fand, \fto, 0, 1}$ be an arbitrary complete residuated lattice. 

\subsection{Fuzzy Sets}
A {\em fuzzy subset} of a non-empty set $A$, simply called a {\em fuzzy set}, is any function $f: A \to L$. The value $f(a)$, for $a \in A$, represents the fuzzy degree in which $a$ belongs to the fuzzy set. The {\em support} of $f: A \to L$ is the set $\{a \in A \mid f(a) > 0\}$. A fuzzy set is called {\em empty} and denoted by $\emptyset$ if its support is empty.

Given $\{a_1,\ldots,a_n\} \subseteq A$ and $\{v_1,\ldots,v_n\} \subseteq L$, by $\{a_1:v_1$, \ldots, $a_n:v_n\}$ we denote the fuzzy set \mbox{$f : A \to L$} such that $f(a_i) = v_i$ for $1 \leq i \leq n$ and $f(a) = 0$ for $a \in A \setminus \{a_1,\ldots,a_n\}$. 
Similarly, given $\{a_i \mid i \in I\} \subseteq A$ and $\{v_i \mid i \in I\} \subseteq L$, by $\{a_i:v_i \mid i \in I\}$ we denote the fuzzy set \mbox{$f : A \to L$} such that $f(a_i) = v_i$ for $i \in I$ and $f(a) = 0$ for $a \in A \setminus \{a_i \mid i \in I\}$. 

Given fuzzy sets $f, g: A \to L$, we say that $f$ is {\em greater than or equal to}~$g$, denoted by $f \geq g$ or $g \leq f$, if $g(a) \leq f(a)$ for all $a \in A$. We write $g < f$ to denote that $g \leq f$ and $g \neq f$.
The {\em fuzzy degree of that $g$ is a subset of $f$}, denoted by $S(g,f)$, is defined as follows:
\[ S(g,f) = \bigwedge_{a \in A} (g(a) \fto f(a)). \] 
The {\em fuzzy degree of that $g$ is equal to $f$}, denoted by $E(g,f)$, is defined as follows: 
\[ E(g,f) = \bigwedge_{a \in A} (g(a) \fequiv f(a)). \]

A {\em fuzzy relation} $\varphi$ between non-empty sets $A$ and $B$ is any fuzzy subset of the Cartesian product of $A$ and $B$. It is {\em image-finite} if, for every $a \in A$, the set $\{b \in B \mid \varphi(a,b) > 0\}$ is finite. 
A fuzzy relation between $A$ and itself is called a fuzzy relation on~$A$. By $id_A$ we denote the identity fuzzy relation on $A$, which is defined as $\{\tuple{a,a}\!:\!1 \mid a \in A\}$.

The {\em inverse} of a fuzzy relation $\varphi: A \times B \to L$ is the fuzzy relation \mbox{$\varphi^{-1} : B \times A \to L$} defined by $\varphi^{-1}(b,a) = \varphi(a,b)$, for all $a\in A$ and $b \in B$.

Given fuzzy relations \mbox{$\varphi: A \times B \to L$}, \mbox{$\psi: B \times C \to L$} and fuzzy sets $f: A \to L$, $g: B \to L$, we define the {\em compositions} \mbox{$\varphi \circ \psi : A \times C \to L$}, \mbox{$(f \circ \varphi) : B \to L$} and \mbox{$(\varphi \circ g) : A \to L$} as follows:
\begin{eqnarray*}
(\varphi \circ \psi)(a,c) & = & \bigvee_{b \in B} (\varphi(a,b) \fand \psi(b,c)) \\
(f \circ \varphi)(b) & = & \bigvee_{a \in A} (f(a) \fand \varphi(a,b)) \\
(\varphi \circ g)(a) & = & \bigvee_{b \in B} (\varphi(a,b) \fand g(b)),
\end{eqnarray*}
for all $a \in A$, $b \in B$ and $c \in C$. Note that the composition operator $\circ$ is associative, and we have $(\varphi \circ \psi)^{-1} = \psi^{-1} \circ \varphi^{-1}$.

A fuzzy relation $\varphi$ on $A$ is 
{\em reflexive} if $id_A \leq \varphi$, 
{\em symmetric} if $\varphi = \varphi^{-1}$, 
and {\em transitive} if $\varphi \circ \varphi \leq \varphi$. 
If a fuzzy relation is reflexive and transitive, then  it is called a {\em fuzzy preorder}. 
Also, if a fuzzy relation is reflexive, symmetric and transitive, then  it is called a  {\em fuzzy equivalence}.

Given a family $\Phi$ of fuzzy relations between $A$ and $B$, the fuzzy relations $\bigvee\!\Phi$ and $\bigwedge\!\Phi$ between $A$ and $B$ are specified by: 
\[
({\textstyle\bigvee}\Phi)(a,b) = \bigvee_{\varphi \in \Phi}\!\varphi(a,b), \qquad\qquad
({\textstyle\bigwedge}\Phi)(a,b) = \bigwedge_{\varphi \in \Phi}\!\varphi(a,b).
\]

Let $\Phi = (\varphi_n)_{n \in \NN}$ and $\Phi' = (\varphi'_n)_{n \in \NN}$ be sequences of fuzzy relations between $A$ and $B$, and $\Psi = (\psi_n)_{n \in \NN}$ a sequence of fuzzy relations between $B$ and $C$. We say that $\Phi$ is {\em less than or equal to} $\Phi'$, denoted by $\Phi \leq \Phi'$, if $\varphi_n \leq \varphi'_n$ for all $n \in \NN$. By $\Phi \circ \Psi$ we denote the sequence $(\varphi_n \circ \psi_n)_{n \in \NN}$ of fuzzy relations between $A$ and $C$, and by $\Phi^{-1}$ the sequence $(\varphi_n^{-1})_{n \in \NN}$ of fuzzy relations between $B$ and $A$. 
We treat a sequence $\Phi = (\varphi_n)_{n \in \NN}$ as the function that maps each natural number $n$ to $\varphi_n$. 
Given a family $\mathbf{\Phi}$ of sequences of fuzzy relations between $X$ and $Y$, by $\bigvee\!\mathbf{\Phi}$ we denote the sequence $\big(\bigvee\{\Phi(n) \mid \Phi \in \mathbf{\Phi}\}\big)_{n \in \NN}$.

%===============================================================================

\subsection{Fuzzy Automata}

A {\em fuzzy automaton} over an alphabet $\Sigma$ (and a complete residuated lattice $\mL$) is a structure $\mA = \tuple{A, \deltaA, \sigmaA, \tauA}$, where $A$ is a non-empty set of states, $\delta^\mA : A \times \Sigma \times A \to L$ is the fuzzy transition function, $\sigmaA : A \to L$ is the fuzzy set of initial states, and $\tauA : A \to L$ is the fuzzy set of terminal states. 
For $s \in \Sigma$, by $\deltaA_s$ we denote the fuzzy relation on $A$ specified by $\deltaA_s(x,y) = \deltaA(x,s,y)$. 
For $x \in A$, by $\mA_x$ we denote the fuzzy automaton that differs from $\mA$ only in that $\sigma^{\mA_x} = \{x:1\}$. 

A fuzzy automaton $\mA = \tuple{A, \deltaA, \sigmaA, \tauA}$ is said to be {\em image-finite} if, for every $s \in \Sigma$ and $x \in A$, the set $\{y \in A \mid \deltaA_s(x,y) > 0\}$ is finite and, furthermore, the support of $\sigmaA$ is also finite.

A fuzzy subset of $\Sigma^*$ is called a \emph{fuzzy language} over the alphabet $\Sigma$ (and $\mL$). The \emph{fuzzy language recognized by a fuzzy automaton} $\mA = \tuple{A,\deltaA,\sigmaA,\tauA}$ is the fuzzy language $\bL(\mA)$ over $\Sigma$ specified by:
\[ \bL(\mA)(s_1 s_2 \ldots s_n)=\sigmaA\circ\deltaA_{s_1}\circ\deltaA_{s_2}\circ\ldots\circ\deltaA_{s_n}\circ\tauA, \]
for $n \geq 0$ and $s_1,\ldots,s_n \in \Sigma$. 
By $\bLdb(\mA)$ we denote the restriction of $\bL(\mA)$ to words of length not greater than $n$. That is, for $w \in \Sigma^*$, $\bLdb(\mA)(w)$ is $\bL(\mA)(w)$ if $|w| \leq n$, and 0 otherwise. 
We call $\bLdb(\mA)$ the fuzzy language consisting of words recognized by $\mA$ with a length bounded by $n$ (or a fuzzy length-bounded language recognized by $\mA$, for short).

Given fuzzy automata $\mA = \tuple{A, \deltaA, \sigmaA, \tauA}$ and $\mAp = \tuple{A', \deltaAp, \sigmaAp, \tauAp}$ and a fuzzy relation \mbox{$\varphi: A \times A' \to L$}, we denote 
\begin{eqnarray}
\nZs & = & S(\sigma^\mA, \sigmaAp \circ \varphi^{-1}), \label{eq.FuzzySimNorm}\\
\nZbs & = & \nZs \land \normS{\varphi^{-1}}{\mAp}{\mA}. \label{eq.FuzzyBiSimNorm}
\end{eqnarray}

From now on, if not stated otherwise, let $\mA = \tuple{A, \deltaA, \sigmaA, \tauA}$ and $\mAp = \tuple{A', \deltaAp, \sigmaAp, \tauAp}$ be arbitrary fuzzy automata over an alphabet~$\Sigma$. 

\subsection{Fuzzy Simulations and Bisimulations between Fuzzy Automata}

A fuzzy relation \mbox{$\varphi: A \times A' \to L$} is called a {\em fuzzy simulation} between fuzzy automata~$\mA$ and~$\mAp$~\cite{FB4FA} if it satisfies the following conditions:
\begin{eqnarray}
\varphi^{-1} \circ \tauA & \leq & \tauAp \label{eq: HFHAJ 3} \\
\varphi^{-1} \circ \deltaA_s & \leq & \deltaAp_s \circ \varphi^{-1}\quad \textrm{for all $s \in \Sigma$.} \label{eq: HFHAJ 2}
\end{eqnarray}
A fuzzy simulation between $\mA$ and itself is called a~{\em fuzzy auto-simulation} of~$\mA$.  
The {\em norm} of a fuzzy simulation $\varphi$ between $\mA$ and $\mAp$ is defined to be $\nZs$, which is specified by~\eqref{eq.FuzzySimNorm}.

Moreover, a fuzzy relation \mbox{$\varphi: A \times A' \to L$} is called a~{\em fuzzy bisimulation} between fuzzy automata $\mA$ and $\mAp$ \cite{FB4FA} if it is a fuzzy simulation between $\mA$ and $\mAp$ and its inverse $\varphi^{-1}$ is a fuzzy simulation between~$\mAp$ and~$\mA$. In other words, we say that $\varphi$ is a fuzzy bisimulation between $\mA$ and $\mAp$ iff it satisfies the conditions \eqref{eq: HFHAJ 3}, \eqref{eq: HFHAJ 2} and the following ones:
\begin{eqnarray}
\varphi \circ \tauAp & \leq & \tauA \label{eq: HFHAJ 6} \\
\varphi \circ \deltaAp_s & \leq & \deltaA_s \circ \varphi\quad \textrm{ for all $s \in \Sigma$.} \label{eq: HFHAJ 5}
\end{eqnarray}
A fuzzy bisimulation between $\mA$ and itself is called a~{\em fuzzy auto-bisimulation} of~$\mA$. 
The {\em norm} of a fuzzy bisimulation $\varphi$ between $\mA$ and $\mAp$ is defined to be $\nZbs$, which is specified by~\eqref{eq.FuzzyBiSimNorm}.

\begin{example}\label{example: LJDNS 0}
Let the underlying lattice $\tuple{L, \leq, 0, 1}$ be the unit interval $[0,1]$ with the usual order. 
Consider the fuzzy automata $\mA = \tuple{A, \deltaA, \sigmaA, \tauA}$ and $\mAp = \tuple{A', \deltaAp, \sigmaAp, \tauAp}$ over the alphabet $\Sigma = \{s\}$ depicted in Figure~\ref{fig: HDBKA} (on page~\pageref{fig: HDBKA}) and specified below:
\begin{itemize}
\item $A = \{u,v\}$, $\sigmaA = \{u\!:\!1\}$, $\tauA = \{v\!:\!1\}$ and $\deltaA_s = \{\tuple{u,v}\!:\!0.4, \tuple{v,v}\!:\!0.5\}$,
\item $A' = \{u',v'\}$, $\sigmaAp = \{u'\!:\!1\}$, $\tauAp = \{v'\!:\!0.8\}$ and $\deltaAp_s = \{\tuple{u',v'}\!:\!0.5, \tuple{v',v'}\!:\!0.4\}$.
\end{itemize}
For the cases where $\fand$ is the G\"odel, {\L}ukasiewicz or product t-norm, we can verify that the fuzzy relation $\varphi : A \times A' \to L$ (respectively, $\psi : A \times A' \to L$) specified below is a fuzzy simulation (respectively, bisimulation) between $\mA$ and~$\mAp$. 

\begin{center}
    \begin{NiceTabular}{@{}p{6em}p{17em}p{17em}@{}}[]
        \toprule
        \textbf{T-norm ($\fand$)} & \textbf{Fuzzy simulation ($\varphi$)} & \textbf{Fuzzy bisimulation ($\psi$)} \\
        \midrule
        G\"odel & $\varphi = \{\tuple{u,u'}\!:\!1, \tuple{u,v'}\!:\!1, \tuple{v,v'}\!:\!0.4\}$ & $\psi = \{\tuple{u,u'}\!:\!0.4, \tuple{v,v'}\!:\!0.4\}$ \\
         & $\nZs = 1$ & $\normBS{\psi}{\mA}{\mAp} = 0.4$ \\
         \midrule
        {\L}ukasiewicz & $\varphi  =  \{\tuple{u,u'}\!:\!0.6, \tuple{u,v'}\!:\!0.6, \tuple{v,v'}\!:\!0.5\}$ & $\psi = \{\tuple{u,u'}\!:\!0.5, \tuple{u,v'}\!:\!0.2, \tuple{v,v'}\!:\!0.5\}$ \\
         & $\nZs = 0.6$ & $\normBS{\psi}{\mA}{\mAp} = 0.5$ \\
         \midrule
        product & $\varphi = \emptyset$ & $ \psi = \emptyset$ \\
         & $\nZs = 0$ & $\normBS{\psi}{\mA}{\mAp} = 0$ \\
        \bottomrule
    \end{NiceTabular}
\end{center}

By Examples~\ref{example: LJDNS 1} and~\ref{example: LJDNS 2} and Corollaries~\ref{cor: HGRKW} and~\ref{cor: HGRKW 2}, which are given in the further parts of this article, we can further verify that $\varphi$ (respectively, $\psi$) is the greatest such fuzzy simulation (respectively, bisimulation). 
\myend
\end{example}

%===============================================================================

\section{Depth-Bounded Fuzzy Simulations between Fuzzy Automata}
\label{section: fs-4-fa}

In this section, we formally define depth-bounded fuzzy simulations between fuzzy automata and study their properties. In particular, we investigate their relationship with fuzzy simulations, the fuzzy preservation of the fuzzy length-bounded languages recognized by a fuzzy automaton under depth-bounded fuzzy simulations, as well as the Hennessy-Milner property of such simulations.

A~{\em depth-bounded fuzzy simulation} between fuzzy automata $\mA$ and $\mAp$ is a sequence $(\varphi_n)_{n \in \NN}$ of fuzzy relations \mbox{$\varphi_n: A \times A' \to L$} such that:
\begin{eqnarray}
\varphi_n & \leq & \varphi_{n-1}\quad \textrm{for all $n \geq 1$} \label{eq: HFKJA 1} \\
\varphi_0^{-1} \circ \tauA & \leq & \tauAp \label{eq: HFKJA 2} \\
\varphi_n^{-1} \circ \deltaA_s & \leq & \deltaAp_s \circ \varphi_{n-1}^{-1}\quad \textrm{for all $s \in \Sigma$ and $n \geq 1$.} \label{eq: HFKJA 3}
\end{eqnarray}

The condition~\eqref{eq: HFKJA 1} states that the sequence $(\varphi_n)_{n \in \NN}$ is decreasing. We will use the word \enquote{decreasing} with this sense instead of \enquote{non-increasing} because the underlying residuated lattice is not required to be linear. 

The {\em norm} of a depth-bounded fuzzy simulation $\Phi = (\varphi_n)_{n \in \NN}$ between $\mA$ and $\mAp$ is defined to be: 
\begin{equation}
\nZdbs\ =\ \bigwedge_{n \in \NN} \normS{\varphi_n}{\mA}{\mAp}. \label{eq: JDJAKJ}
\end{equation}
In addition, a depth-bounded fuzzy simulation between $\mA$ and itself is called a~{\em depth-bounded fuzzy auto-simulation} of~$\mA$.

\begin{figure}
\begin{center}
\begin{tikzpicture}[->,>=stealth,auto]
    \tikzset{every state/.style={inner sep=0.12cm,minimum size=0.8cm}}
    \tikzstyle{every node}=[font=\footnotesize]
\node (A) {$\mA$};
\node[state] (u) [node distance=1.3cm, below of=A] {$u$};
\node[state] (v) [node distance=2.0cm, below of=u] {$v$};
\node[left of=u] (in) {};
\node[left of=v] (out) {};
\path[->](in) edge[above, sloped] node{1} (u);
\path[->](v) edge[above, sloped] node{1} (out);
\draw (u) to node[right]{0.4} (v);
\draw (v) edge[loop below,out=-120,in=-60,looseness=5] node{$0.5$} (v);
\node (Ap) [node distance=4cm, right of=A] {$\mAp$};
\node[state] (up) [node distance=1.3cm, below of=Ap] {$u'$};
\node[state] (vp) [node distance=2.0cm, below of=up] {$v'$};
\node[left of=up] (inp) {};
\node[left of=vp] (outp) {};
\path[->](inp) edge[above, sloped] node{1} (up);
\path[->](vp) edge[above, sloped] node{0.8} (outp);
\draw (up) to node[right]{0.5} (vp);
\draw (vp) edge[loop below,out=-120,in=-60,looseness=5] node{$0.4$} (vp);
\end{tikzpicture}
\caption{An illustration for Examples~\ref{example: LJDNS 1} and~\ref{example: LJDNS 2}.\label{fig: HDBKA}}
\end{center}
\end{figure}
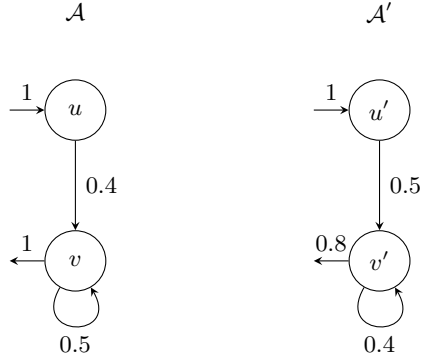

\begin{example}\label{example: LJDNS 1}
Let $\tuple{L, \leq}$, $\Sigma$, $\mA$ and $\mAp$ be as in Example~\ref{example: LJDNS 0}. Recall that the fuzzy automata $\mA$ and $\mAp$ are depicted in Figure~\ref{fig: HDBKA} (on page~\pageref{fig: HDBKA}). 
\begin{center}
    \begin{NiceTabular}{@{}p{6em}p{27em}p{7em}@{}}[]
        \toprule
        \textbf{T-norm} & \textbf{Depth-bounded fuzzy simulation} & \textbf{Norm}  \\
        \textbf{($\fand$)} & \textbf{($\Phi = (\varphi_n)_{n \in \NN}$)} & \textbf{($\nZdbs$)}  \\
        \midrule
        G\"odel & $\varphi_0 = \{\tuple{u,u'}\!:\!1, \tuple{u,v'}\!:\!1, \tuple{v,v'}\!:\!0.8\}$ & $\nZdbs = 1$ \\
                & $\varphi_n = \{\tuple{u,u'}\!:\!1, \tuple{u,v'}\!:\!1, \tuple{v,v'}\!:\!0.4\}$ for all $n \geq 1$ &  \\
         \midrule
        {\L}ukasiewicz & $\varphi_0 = \{\tuple{u,u'}\!:\!1, \tuple{u,v'}\!:\!1, \tuple{v,v'}\!:\!0.8\}$ & $\nZdbs = 0.6$ \\
	             & $\varphi_1 = \{\tuple{u,u'}\!:\!0.9, \tuple{u,v'}\!:\!0.8, \tuple{v,v'}\!:\!0.7\}$ &\\
	           & $\varphi_2 = \{\tuple{u,u'}\!:\!0.8, \tuple{u,v'}\!:\!0.7, \tuple{v,v'}\!:\!0.6\}$ &\\
	             & $\varphi_3 = \{\tuple{u,u'}\!:\!0.7, \tuple{u,v'}\!:\!0.6, \tuple{v,v'}\!:\!0.5\}$ &\\
	           & $\varphi_n = \{\tuple{u,u'}\!:\!0.6, \tuple{u,v'}\!:\!0.6, \tuple{v,v'}\!:\!0.5\}$ for all $n \geq 4$ &
 	 \\
         \midrule
        product & $\varphi_0 = \{\tuple{u,u'}\!:\!1, \tuple{u,v'}\!:\!1, \tuple{v,v'}\!:\!0.8\}$ & $\nZdbs = 0$  \\
	        & $\varphi_n = \{\tuple{u,u'}\!:\!0.8^{n-1}, \tuple{u,v'}\!:\!0.8^n, \tuple{v,v'}\!:\!0.8^{n+1}\}$ for all $n \geq 1$ & 
 	\\
        \bottomrule
    \end{NiceTabular}
\end{center}

For the cases where $\fand$ is the G\"odel, {\L}ukasiewicz or product t-norm, it can be checked that the sequence $\Phi = (\varphi_n)_{n \in \NN}$ specified above is the greatest depth-bounded fuzzy simulation between $\mA$ and~$\mAp$. 
\myend
\end{example}

%\subsection{Basic Properties}

\begin{proposition}\label{prop: JHFNW}
Let $\mA$, $\mAp$ and $\mAdp$ be fuzzy automata. Furthermore, let $\mathbf{\Phi}$ be a family of depth-bounded fuzzy simulations between $\mA$ and $\mAp$, and $\Psi$ a depth-bounded fuzzy simulation between $\mAp$ and $\mAdp$. Then, the following properties hold:
\begin{enumerate}[(a)]
\item If $\Phi_1, \Phi_2 \in \mathbf{\Phi}$ and $\Phi_1 \le \Phi_2$, then 
\[ \normS{\Phi_1}{\mA}{\mAp} \leq \normS{\Phi_2}{\mA}{\mAp}. \]
\item If $\Phi \in \mathbf{\Phi}$, then the composition $\Phi \circ \Psi$ is a depth-bounded fuzzy simulation between $\mA$ and $\mAdp$ such that 
\begin{eqnarray}\label{eq: HFJHV}
\normS{\Phi}{\mA}{\mAp} \fand \normS{\Psi}{\mAp}{\mAdp} \leq \normS{\Phi \circ \Psi}{\mA}{\mAdp}.
\end{eqnarray}
\item The meet $\bigvee\!\mathbf{\Phi}$ is also a depth-bounded fuzzy simulation between $\mA$ and $\mAp$. 
\item There always exists the greatest depth-bounded fuzzy simulation between $\mA$ and~$\mAp$.
\end{enumerate}
\end{proposition}

\begin{proof}
\begin{enumerate}[(a)]
\item It follows directly from~\eqref{fop: GDJSK 10} and~\eqref{fop: GDJSK 20}.
	
\item Let $\Phi = (\varphi_n)_{n \in \NN} \in \mathbf{\Phi}$ and $\Psi = (\psi_n)_{n \in \NN}$. Then, for every $s \in \Sigma$ and $n \geq 1$, we have:
\[
\begin{array}{rclrcl}
\varphi_n & \leq & \varphi_{n-1} &
\psi_n & \leq & \psi_{n-1} \\
\varphi_0^{-1} \circ \tauA & \leq & \tauAp & 
\psi_0^{-1} \circ \tauAp & \leq & \tauAdp \\ 
\varphi_n^{-1} \circ \deltaA_s & \leq & \deltaAp_s \circ \varphi_{n-1}^{-1}\qquad & 
\psi_n^{-1} \circ \deltaAp_s & \leq & \deltaAdp_s \circ \psi_{n-1}^{-1}.
\end{array}
\]
As  $\fand$ is monotonic and $\circ$ is associative, it follows that, for every $s \in \Sigma$ and $n \geq 1$, 
%Since $\circ$ is associative and $\fand$ is monotonic (as shown by~\eqref{fop: GDJSK 10}), for every $s \in \Sigma$ and $n \geq 1$, 
\begin{eqnarray*}
& \varphi_n \circ \psi_n \leq \varphi_{n-1} \circ \psi_{n-1}, \\ 
& (\varphi_0 \circ \psi_0)^{-1} \circ \tauA  
= \psi_0^{-1} \circ \varphi_0^{-1} \circ \tauA 
\leq \psi_0^{-1} \circ \tauAp  
\leq \tauAdp, \\
& (\varphi_n \circ \psi_n)^{-1} \circ \deltaA_s 
= \psi_n^{-1} \circ \varphi_n^{-1} \circ \deltaA_s
\leq \psi_n^{-1} \circ \deltaAp_s \circ \varphi_{n-1}^{-1}\ \leq \\ 
& \leq\ \deltaAdp_s \circ \psi_{n-1}^{-1} \circ \varphi_{n-1}^{-1} 
= \deltaAdp_s \circ (\varphi_{n-1} \circ \psi_{n-1})^{-1}.
\end{eqnarray*}
Therefore, $\Phi \circ \Psi = (\varphi_n \circ \psi_n)_{n \in \NN}$ is a depth-bounded fuzzy simulation between $\mA$ and $\mAdp$. 
The inequality~\eqref{eq: HFJHV} means
\[ 
\left(\bigwedge_{n \in \NN} \normS{\varphi_n}{\mA}{\mAp}\right) \fand 
	\bigwedge_{n \in \NN} \normS{\psi_n}{\mAp}{\mAdp} \leq
	\bigwedge_{n \in \NN} \normS{\varphi_n \circ \psi_n}{\mA}{\mAdp}.
\]
To prove this, by~\eqref{fop: GDJSK 10}, it suffices to show that, for every $n \in \NN$,  
\[ \normS{\varphi_n}{\mA}{\mAp} \fand 
	\normS{\psi_n}{\mAp}{\mAdp} \leq
	\normS{\varphi_n \circ \psi_n}{\mA}{\mAdp}.
\]
Similarly, it suffices to show that, for every $n \in \NN$ and $x \in A$,  
\[ (\sigma^\mA(x) \fto (\sigmaAp \circ \varphi_n^{-1})(x)) \fand 
	\normS{\psi_n}{\mAp}{\mAdp} \leq
	(\sigma^\mA(x) \fto (\sigmaAdp \circ (\varphi_n \circ \psi_n)^{-1})(x)).
\]
By \eqref{fop: GDJSK 80a} and~\eqref{fop: GDJSK 20}, it suffices to show that, for every $n \in \NN$ and $x \in A$, 
\[ (\sigmaAp \circ \varphi_n^{-1})(x) \fand \normS{\psi_n}{\mAp}{\mAdp} \leq (\sigmaAdp \circ \psi_n^{-1} \circ \varphi_n^{-1})(x). \]
By \eqref{fop: GDJSK 155} and~\eqref{fop: GDJSK 10}, it suffices to show that, for every $n \in \NN$ and $x' \in A'$, 
\[ \sigmaAp(x') \fand \normS{\psi_n}{\mAp}{\mAdp} \leq (\sigmaAdp \circ \psi_n^{-1})(x'). \]
This holds due to~\eqref{fop: GDJSK 00} and the definition of $\normS{\psi_n}{\mAp}{\mAdp}$. 
\item Since every $\Phi \in \mathbf{\Phi}$ is a decreasing sequence, for every $n \geq 1$, 
\[ \textstyle\bigvee\{\Phi(n) \mid \Phi \in \mathbf{\Phi}\} \leq \textstyle\bigvee\{\Phi(n-1) \mid \Phi \in \mathbf{\Phi}\}. \]
By using~\eqref{fop: GDJSK 155} and the assumption about $\mathbf{\Phi}$, we have 
\[
	\big(\textstyle\bigvee\{\Phi(0) \mid \Phi \in \mathbf{\Phi}\}\big)^{-1} \circ \tauA 
	= (\textstyle\bigvee\{\Phi^{-1}(0) \mid \Phi \in \mathbf{\Phi}\}) \circ \tauA 
	= \textstyle\bigvee\{\Phi^{-1}(0) \circ \tauA \mid \Phi \in \mathbf{\Phi}\}
	\leq \tauAp. 
\]
Similarly, we also have that, for every $n \geq 1$ and $s \in \Sigma$, 
\begin{eqnarray*}
\big(\textstyle\bigvee\{\Phi(n) \mid \Phi \in \mathbf{\Phi}\}\big)^{-1} \circ \deltaA_s 
& = & (\textstyle\bigvee\{\Phi^{-1}(n) \mid \Phi \in \mathbf{\Phi}\}) \circ \deltaA_s \\
& = & \textstyle\bigvee\{\Phi^{-1}(n) \circ \deltaA_s \mid \Phi \in \mathbf{\Phi}\} \\
& \leq & \textstyle\bigvee\{\deltaAp_s \circ \Phi^{-1}(n-1) \mid \Phi \in \mathbf{\Phi}\} \\
& = & \deltaAp_s \circ \textstyle\bigvee\{\Phi^{-1}(n-1) \mid \Phi \in \mathbf{\Phi}\} \\
& = & \deltaAp_s \circ \big(\textstyle\bigvee\{\Phi(n-1) \mid \Phi \in \mathbf{\Phi}\}\big)^{-1}. 
\end{eqnarray*}
Therefore, $\bigvee\!\mathbf{\Phi}$ is a depth-bounded fuzzy simulation between $\mA$ and $\mAp$. 
\item It follows directly from (c).
\myend
\end{enumerate}
\end{proof}

In what follows, we determine the conditions under which we can relate depth-bounded fuzzy simulations with fuzzy simulations. To that effort, let us say that a lattice $\tuple{L, \le}$ satisfies the {\em join-meet distributivity law} if for any $B \subseteq L$ and $a \in L$ we have
\begin{equation} \label{eq.JMDL}
    a \vee \left( \bigwedge B \right) = \bigwedge_{b \in B} (a \vee b).
\end{equation}

Note that if $\cal L$ is linear, then it satisfies the join-meet distributivity law. 

The following lemma has been proved in~\cite[Lemma~4.2.]{jcss/CIRIC2010}.

\begin{lemma}\label{lem.NonIncSeqAB}
Suppose that $\mL$ satisfies the join-meet distributivity law. Then, for all decreasing sequences $(a_k)_{k \in \NN} \subseteq L$ and $(b_k)_{k \in \NN} \subseteq L$, the following is satisfied:
\begin{equation}
    \bigwedge_{k \in \NN} (a_k \vee b_k) = \left( \bigwedge_{k \in \NN} a_k \right) \vee \left(\bigwedge_{k \in \NN} b_k \right).
\end{equation}
\end{lemma}

Using the above lemma, we can derive the following one.

\begin{lemma}\label{lem.CompFuzzyRelatMJDL}
Suppose that $\mL$ satisfies the join-meet distributivity law and $\fand$ is continuous. Let $\varphi$ be an image-finite fuzzy relation between $X$ and $Y$, and $(\psi_k)_{k \in \NN}$ a decreasing sequence of fuzzy relations between $Y$ and $Z$. Then:
\begin{equation}\label{eq.CompFuzzyRelatDist}
    \varphi \circ \left( \bigwedge_{k \in \NN} \psi_k \right) = \bigwedge_{k \in \NN} (\varphi \circ \psi_k).
\end{equation}
\end{lemma}

\begin{proof}
    Consider arbitrary $x \in X$ and $z \in Z$. By definitions, we have:
    \begin{equation}\label{eq: JHDJS}
    \left( \bigwedge_{k \in \NN} (\varphi \circ \psi_k) \right)(x, z)
    = \bigwedge_{k \in \NN} (\varphi \circ \psi_k) (x, z)
    = \bigwedge_{k \in \NN} \bigvee_{y \in Y} \varphi(x, y) \otimes \psi_k(y, z).
    \end{equation}
    By the assumption, for every $y \in Y$, $(\psi_k(y, z))_{k \in \NN}$ is a decreasing sequence, therefore $(\varphi(x, y) \otimes \psi_k(y, z))_{k \in \NN}$ is also a decreasing sequence. Since $\varphi$ is image-finite, the set $\{ y \in Y \mid \varphi(x, y) > 0 \}$ is finite. Thus, applying Lemma~\ref{lem.NonIncSeqAB} to the RHS of~\eqref{eq: JHDJS} and then using the fact that $\fand$ is continuous, we obtain:
    \[
    \left( \bigwedge_{k \in \NN} (\varphi \circ \psi_k) \right)(x, z)
    =  \bigvee_{y \in Y} \bigwedge_{k \in \NN} \varphi(x, y) \otimes \psi_k(y, z)
    =  \bigvee_{y \in Y} \varphi(x, y) \otimes \left( \bigwedge_{k \in \NN} \psi_k(y, z) \right).
    \]
    Continuing the above derivation, we obtain:
    \[
    \left( \bigwedge_{k \in \NN} (\varphi \circ \psi_k) \right)(x, z)
    =  \bigvee_{y \in Y} \varphi(x, y) \otimes \left( \bigwedge_{k \in \NN} \psi_k\right)(y, z)
    = \left( \varphi \circ \left( \bigwedge_{k \in \NN} \psi_k\right) \right)(x, z).
    \]
    As this is valid for every $x \in X$ and $z \in Z$, \eqref{eq.CompFuzzyRelatDist} follows and the proof is completed.
\myend
\end{proof}

The following lemma is similar to the above one and can be proved analogously. 

\begin{lemma}\label{lem.CompFuzzyRelatMJDL 2}
Suppose that $\mL$ satisfies the join-meet distributivity law and $\fand$ is continuous. Let $f$ be a fuzzy subset of $Y$ with a finite support and $(\psi_k)_{k \in \NN}$ a decreasing sequence of fuzzy relations between $Y$ and $Z$. Then:
\begin{equation}\label{eq.CompFuzzyRelatDist 2}
    f \circ \left( \bigwedge_{k \in \NN} \psi_k \right) = \bigwedge_{k \in \NN} (f \circ \psi_k).
\end{equation}
\end{lemma}

The following theorem relates depth-bounded fuzzy simulations to fuzzy simulations. Roughly speaking, a depth-bounded fuzzy simulation is a decreasing sequence of fuzzy relations that, under certain light conditions, converges to a fuzzy simulation. 

\begin{theorem}\label{theorem: HDKAK}
Suppose that $\mL$ satisfies the join-meet distributivity law and $\fand$ is continuous. Let $\Phi = (\varphi_n)_{n \in \NN}$ be a depth-bounded fuzzy simulation between fuzzy automata $\mA$ and $\mAp$ such that $\mAp$ is image-finite. Then, the fuzzy relation $\varphi = \bigwedge\!\Phi$ is a fuzzy simulation between $\mA$ and $\mAp$ with $\nZs = \nZdbs$.
\end{theorem}

\begin{proof}
    The fact that $\varphi$ satisfies~\eqref{eq: HFHAJ 3} immediately follows from \eqref{eq: HFKJA 1}, \eqref{eq: HFKJA 2} and \eqref{fop: GDJSK 10}. Let us prove that $\varphi$ also satisfies~\eqref{eq: HFHAJ 2}. First, note that, for every $s \in \Sigma$, we have:
    \[
        \varphi^{-1} \circ \deltaA_s = 
        \left( \bigwedge_{n \in \NN}  \varphi_{n} \right)^{-1} \circ \deltaA_s =
        \left( \bigwedge_{n \in \NN}  \varphi_{n}^{-1} \right) \circ \deltaA_s
        \le \bigwedge_{n \in \NN}  \left( \varphi_{n}^{-1} \circ \deltaA_s \right).
    \]
    Applying~\eqref{eq: HFKJA 3} we further obtain:
    \[
        \varphi^{-1} \circ \deltaA_s 
        \le \bigwedge_{n \geq 1}  \left( \deltaAp_s \circ \varphi_{n - 1}^{-1} \right).
    \]
    Because $\mAp$ is image-finite, applying Lemma~\ref{lem.CompFuzzyRelatMJDL}, we get:
    \[
        \varphi^{-1} \circ \deltaA_s 
        \le \bigwedge_{n \geq 1}  \left( \deltaAp_s \circ \varphi_{n - 1}^{-1} \right)
        = \deltaAp_s \circ \left( \bigwedge_{n \geq 1}  \varphi_{n - 1}^{-1} \right)
        = \deltaAp_s \circ \left( \bigwedge_{n \in \NN}  \varphi_{n}\right)^{-1} 
        =  \deltaAp_s \circ \varphi^{-1}.
    \]
    Thus, we have proved that $\varphi$ satisfies~\eqref{eq: HFHAJ 2}, and it is therefore a fuzzy simulation between~$\mA$ and~$\mAp$. It remains to show that $\nZs = \nZdbs$. Using~\eqref{fop: GDJSK 220}, we have
\begin{eqnarray*}
\nZdbs & = & \bigwedge_{n \in \NN} S(\sigma^\mA, \sigmaAp \circ \varphi_n^{-1}) \\ 
& = & \bigwedge_{n \in \NN} \bigwedge_{x \in A} (\sigma^\mA(x) \fto (\sigmaAp \circ \varphi_n^{-1})(x)) \\
& = & \bigwedge_{x \in A} \bigwedge_{n \in \NN} (\sigma^\mA(x) \fto (\sigmaAp \circ \varphi_n^{-1})(x)) \\
& = & \bigwedge_{x \in A} (\sigma^\mA(x) \fto \bigwedge_{n \in \NN} (\sigmaAp \circ \varphi_n^{-1})(x)).
\end{eqnarray*}
Note that the support of $\sigmaAp$ is finite because $\mAp$ is image-finite. Thus, we can apply Lemma~\ref{lem.CompFuzzyRelatMJDL 2} and continue the above derivation as follows:
\begin{eqnarray*}
\nZdbs & = &  \bigwedge_{x \in A} \left(\sigma^\mA(x) \fto \left(\sigmaAp \circ \left( \bigwedge_{n \in \NN} \varphi_n^{-1}\right) \right)(x)\right) \\
& = & \bigwedge_{x \in A} \left(\sigma^\mA(x) \fto \left(\sigmaAp \circ \varphi^{-1} \right)(x)\right) \\
& = & S( \sigmaA, \sigmaAp \circ \varphi^{-1})\\
& = & \nZs,
\end{eqnarray*}
which completes the proof.
\myend
\end{proof}

The following two examples show that the assumptions that $\fand$ is continuous and $\mAp$ is image-finite are necessary for Theorem~\ref{theorem: HDKAK}.

\begin{figure}
\begin{center}
\begin{tikzpicture}[->,>=stealth,auto]
\tikzset{every state/.style={inner sep=0.12cm,minimum size=0.8cm}}
\tikzstyle{every node}=[font=\footnotesize]
\node (A) {$\mA$};
\node[state] (u) [node distance=1.3cm, below of=A] {$u$};
\node[state] (v) [node distance=2.0cm, below of=u] {$v$};
    \node[left of=u] (in) {};
    \node[left of=v] (out) {};
    \path[->](in) edge[above, sloped] node{1} (u);
    \path[->](v) edge[above, sloped] node{1} (out);
\draw (u) to node[right]{1} (v);
\draw (v) edge[loop below,out=-120,in=-60,looseness=5] node{$1$} (v);
\node (Ap) [node distance=3cm, right of=A] {$\mAp$};
\node[state] (up) [node distance=1.3cm, below of=Ap] {$u'$};
\node[state] (vp0) [node distance=2.0cm, below of=up] {$v'_0$};
\node[state] (vp1) [node distance=2.5cm, right of=vp0] {$v'_1$};
\node[state] (vp2) [node distance=2.5cm, right of=vp1] {$v'_2$};
\node (vp3) [node distance=1.5cm, right of=vp2] {\ldots};
\node (avp2) [node distance=1.2cm, above of=vp2] {};
\node (lavp2) [node distance=1.2cm, left of=avp2] {\ldots};
\draw (up) to node[left]{1} (vp0);
\draw (up) to node[left, shift={(0,-2pt)}]{1} (vp1);
\draw (up) to node[right, shift={(-5pt,6pt)}]{1} (vp2);
\draw (vp1) to node[below, shift={(2pt,0)}]{1} (vp0);
\draw (vp2) to node[below, shift={(2pt,0)}]{1} (vp1);
    \node[left of=up] (inp) {};
    \node[below of=vp0] (outp0) {};
    \node[below of=vp1] (outp1) {};
    \node[below of=vp2] (outp2) {};
    \path[->](inp) edge[above, sloped] node{1} (up);
    \path[->](vp0) edge[right] node{1} (outp0);
    \path[->](vp1) edge[right] node{1} (outp1);
    \path[->](vp2) edge[right] node{1} (outp2);
\end{tikzpicture}
\caption{An illustration for Example~\ref{example: OFJAH 1}.\label{fig: GDJWJ}}
\end{center}
\end{figure}
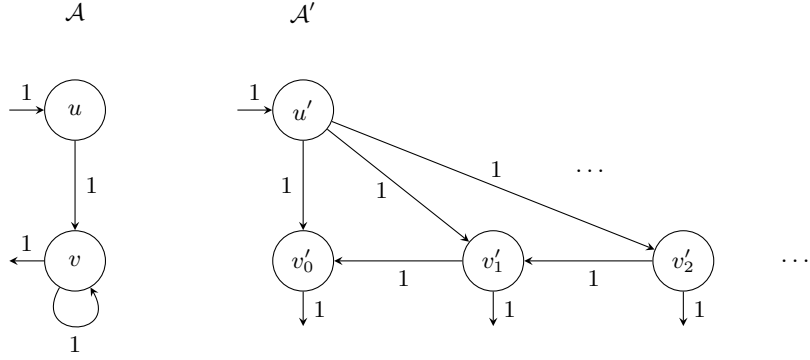

\begin{example}\label{example: OFJAH 1}
Consider the fuzzy automata $\mA = \tuple{A, \deltaA, \sigmaA, \tauA}$ and $\mAp = \tuple{A', \deltaAp, \sigmaAp, \tauAp}$ over the alphabet $\Sigma = \{s\}$ depicted in Figure~\ref{fig: GDJWJ} and specified below:
\begin{itemize}
\item $A = \{u,v\}$, $\sigmaA = \{u\!:\!1\}$, $\tauA = \{v\!:\!1\}$ and $\deltaA_s = \{\tuple{u,v}\!:\!1, \tuple{v,v}\!:\!1\}$,
\item $A' = \{u',v'_i \mid i \in \NN\}$, $\sigmaAp = \{u'\!:\!1\}$, $\tauAp = \{v'_i\!:\!1 \mid i \in \NN\}$ and\\ 
$\deltaAp_s = \{\tuple{u',v'_i}\!:\!1, \tuple{v'_{i+1},v'_i}\!:\!1 \mid i \in \NN\}$.
\end{itemize}
For $n \in \NN$, let $\varphi_n = \{\tuple{u,u'}\!:\!1, \tuple{u,v'_i}\!:\!1, \tuple{v,v'_i}\!:\!1 \mid i \geq n\}$ be a fuzzy relation between $A$ and $A'$. It is easy to check that $\Phi = (\varphi_n)_{n \in \NN}$ is the greatest depth-bounded fuzzy simulation between~$\mA$ and~$\mAp$. Let $\varphi = \bigwedge\!\Phi = \{\tuple{u,u'}\!:\!1\}$. It is easy to see that $\varphi$ does not satisfy the condition~\eqref{eq: HFHAJ 2} and is therefore not a fuzzy simulation between $\mA$ and $\mAp$. Theorem~\ref{theorem: HDKAK} is not applicable to this case because $\mAp$ is not image-finite.
\myend
\end{example}

\begin{example}\label{example: OFJAH 2}
Let $\mL = \tuple{L, \leq, \fand, \fto, 0, 1}$ be any linear complete residuated lattice such that $\fand$ is not continuous.\footnote{We do not have any concrete example of such a lattice, but we assume that such lattices exist, as the condition of being continuous with respect to infima does not seem to follow from the axioms of linear complete residuated lattices.} Thus, there exist $a \in L$ and an infinite decreasing chain $b_0 > b_1 > b_2 > \ldots$ of elements of $L$ such that 
\[
	a \fand \bigwedge_{n \in \NN}\! b_n \ \neq\ \bigwedge_{n \in \NN}\! (a \fand b_n), 
\]
which means
\[
a \fand \bigwedge_{n \in \NN}\! b_n \ <\ \bigwedge_{n \in \NN}\! (a \fand b_n).
\]
Let $b = \bigwedge_{n \in \NN} b_n$ and let $c$ be the value of the RHS of the above inequality. We have $a \fand b < c \leq a \fand b_n$ for all $n \in \NN$. 
Consider the fuzzy automata $\mA = \tuple{A, \deltaA, \sigmaA, \tauA}$ and $\mAp = \tuple{A', \deltaAp, \sigmaAp, \tauAp}$ over the alphabet $\Sigma = \{s\}$ with
\begin{itemize}
\item $A = \{u,v\}$, $\sigmaA = \{u\!:\!1\}$, $\tauA = \{v\!:\!1\}$ and $\deltaA_s = \{\tuple{u,v}\!:\!1\}$,
\item $A' = \{u',v'\}$, $\sigmaAp = \{u'\!:\!1\}$, $\tauAp = \{v'\!:\!1\}$ and $\deltaAp_s = \{\tuple{u',v'}\!:\!a\}$.
\end{itemize}
For $n \in \NN$, let $\varphi_n = \{\tuple{u,u'}\!:\!c, \tuple{v,v'}\!:\!b_n\}$ be a fuzzy relation between $A$ and $A'$. It is easy to check that $\Phi = (\varphi_n)_{n \in \NN}$ is a depth-bounded fuzzy simulation between~$\mA$ and~$\mAp$. Let $\varphi = \bigwedge\!\Phi = \{\tuple{u,u'}\!:\!c, \tuple{v,v'}\!:\!b\}$. It is easy to see that $\varphi$ does not satisfy the condition~\eqref{eq: HFHAJ 2} and is therefore not a fuzzy simulation between $\mA$ and $\mAp$. Theorem~\ref{theorem: HDKAK} is not applicable to this case because $\fand$ is not continuous.
\myend
\end{example}

The following corollary states that, under certain light conditions, the greatest depth-bounded fuzzy simulation between two fuzzy automata is a convergent sequence of fuzzy relations that approximate the greatest fuzzy simulation between the fuzzy automata.

\begin{corollary}\label{cor: HGRKW}
Let $\Phi = (\varphi_n)_{n \in \NN}$ be the greatest depth-bounded fuzzy simulation between fuzzy automata $\mA$ and $\mAp$. If $\mL$ satisfies the join-meet distributivity law, $\fand$ is continuous and $\mAp$ is image-finite, then $\bigwedge\!\Phi$ is the greatest fuzzy simulation between $\mA$ and $\mAp$.
\end{corollary}

\begin{proof}
Assuming that the premises of the assertion hold, we can apply Theorem~\ref{theorem: HDKAK} and conclude that $\bigwedge\!\Phi$ is a fuzzy simulation between $\mA$ and $\mAp$. Let $\varphi$ be the greatest fuzzy simulation between $\mA$ and $\mAp$ (it always exists~\cite{FB4FA}). Thus, $\bigwedge\!\Phi \leq \varphi$. On the other hand, the sequence that consists of infinitely many $\varphi$ is clearly a depth-bounded fuzzy simulation between $\mA$ and $\mAp$. Hence, $\varphi \leq \varphi_n$ for all $n \in \NN$ and, consequently, $\varphi \leq \bigwedge\!\Phi$. Therefore, $\bigwedge\!\Phi = \varphi$.
\myend
\end{proof}

Example~\ref{example: OFJAH 1} shows that the premise ``$\mAp$ is image-finite'' is necessary for Corollary~\ref{cor: HGRKW}, because the sequence $(\varphi_n)_{n \in \NN}$ used in this example is the greatest depth-bounded fuzzy simulation between~$\mA$ and~$\mAp$. We do not have any example to show that Corollary~\ref{cor: HGRKW} can be strengthened, e.g., by removing the premise ``$\fand$ is continuous'' (the problem remains open).

\begin{corollary}\label{cor: HRKQA}
Let $\Phi = (\varphi_n)_{n \in \NN}$ be the greatest depth-bounded fuzzy auto-simulation of $\mA$. If $\mL$ satisfies the join-meet distributivity law, $\fand$ is continuous and $\mA$ is image-finite, then 
\begin{enumerate}[(a)]
\item $\bigwedge\!\Phi$ is the greatest fuzzy auto-simulation of $\mA$,
\item each $\varphi_n$ is a fuzzy pre-order,
\item $\normS{\Phi}{\mA}{\mA} = 1$.
\end{enumerate}
\end{corollary}

\begin{proof}
\begin{enumerate}[(a)]
\item It follows from Corollary~\ref{cor: HGRKW}.
\item Note that the sequence $\Psi$ that consists of infinite many $id_A$ is a depth-bounded fuzzy auto-simulation of $\mA$. Hence, each $\varphi_n$ is reflexive. By Proposition~\ref{prop: JHFNW}, $\Phi \circ \Phi \leq \Phi$. Hence, each $\varphi_n$ is also transitive. Therefore, each $\varphi_n$ is a fuzzy pre-order.
\item Note that $\normS{\Psi}{\mA}{\mA} = 1$. By Proposition~\ref{prop: JHFNW}, $\normS{\Psi}{\mA}{\mA} \leq \normS{\Phi}{\mA}{\mA}$. Hence, $\normS{\Phi}{\mA}{\mA} = 1$.
\myend
\end{enumerate}
\end{proof}

%\subsection{Preservation of the Recognized Language}

The following theorem states a kind of fuzzy preservation of the fuzzy length-bounded languages recognized by a fuzzy automaton under depth-bounded fuzzy simulations. If $(\varphi_n)_{n \in \NN}$ is a depth-bounded fuzzy simulation between $\mA$ and $\mAp$, then for $\tuple{x,x'} \in A \times A'$, the fuzzy degree in which $\bLdb(\mA_x)$ is a subset of $\bLdb(\mAp_{x'})$ is greater than or equal to $\varphi_n(x,x')$ and, furthermore, the fuzzy degree in which $\bLdb(\mA)$ is a subset of $\bLdb(\mAp)$ is greater than or equal to $\normS{\varphi_n}{\mA}{\mAp}$.

\begin{theorem}\label{theorem: JHFLS}
Let $\Phi = (\varphi_n)_{n \in \NN}$ be a depth-bounded fuzzy simulation between fuzzy automata~$\mA$ and~$\mAp$. Then, for every $n \in \NN$ and every $\tuple{x,x'} \in A \times A'$:
\begin{eqnarray}
\varphi_n(x,x') & \leq & S(\bLdb(\mA_x), \bLdb(\mAp_{x'})), \label{eq: JHFLS 1} \\[0.5ex]
\normS{\varphi_n}{\mA}{\mAp} & \leq & S(\bLdb(\mA), \bLdb(\mAp)). \label{eq: JHFLS 2}
\end{eqnarray}
\end{theorem}

\begin{proof}
In order to prove~\eqref{eq: JHFLS 1}, it is sufficient to show that, for every $u \in \Sigma^*$ with $|u| \leq n$, 
	\begin{equation}
		\varphi_n(x,x') \leq (\bL(\mA_x)(u) \fto \bL(\mAp_{x'})(u)), \label{eq: FSKWN}
	\end{equation}
	or equivalently,
	\[ \varphi_n(x,x') \fand \bL(\mA_x)(u) \leq \bL(\mAp_{x'})(u). \]
	Let $u = s_1\ldots s_k$, with $0 \leq k \leq n$. 
	The above inequality is equivalent to
	\[ \varphi_n(x,x') \fand (\deltaA_{s_1}\circ\ldots\circ\deltaA_{s_k}\circ\tauA)(x) \leq (\deltaAp_{s_1}\circ\ldots\circ\deltaAp_{s_k}\circ\tauAp)(x'). \]
	It suffices to show that 
	\[ \varphi_n^{-1}\circ\deltaA_{s_1}\circ\ldots\circ\deltaA_{s_k}\circ\tauA \leq \deltaAp_{s_1}\circ\ldots\circ\deltaAp_{s_k}\circ\tauAp. \]
	Since $\circ$ is associative, by~\eqref{eq: HFKJA 3}, \eqref{eq: HFKJA 2}, \eqref{eq: HFKJA 1} and~\eqref{fop: GDJSK 10}, we have 
	\begin{eqnarray*}
		& & \varphi_n^{-1}\circ\deltaA_{s_1}\circ\ldots\circ\deltaA_{s_k}\circ\tauA \\
		& \leq & \deltaAp_{s_1}\circ\varphi_{n-1}^{-1}\circ\deltaA_{s_2}\circ\ldots\circ\deltaA_{s_k}\circ\tauA \\
		& \leq & \deltaAp_{s_1}\circ\deltaAp_{s_2}\circ\varphi_{n-2}^{-1}\circ\deltaA_{s_3}\circ\ldots\circ\deltaA_{s_k}\circ\tauA \\
		& \leq & \ldots \\
		& \leq & \deltaAp_{s_1}\circ\deltaAp_{s_2}\circ\ldots\circ\deltaAp_{s_k}\circ\varphi_{n-k}^{-1}\circ\tauA \\
		& \leq & \deltaAp_{s_1}\circ\deltaAp_{s_2}\circ\ldots\circ\deltaAp_{s_k}\circ\varphi_{0}^{-1}\circ\tauA \\
		& \leq & \deltaAp_{s_1}\circ\ldots\circ\deltaAp_{s_k}\circ\tauAp.
	\end{eqnarray*}

Let us now prove~\eqref{eq: JHFLS 2}. Starting from the LHS of~\eqref{eq: JHFLS 2}, by the definition of the norm, we get:
\begin{equation}\label{eq.NormVarphinmAmAp}
\normS{\varphi_n}{\mA}{\mAp} = \bigwedge_{x \in A} \left(\sigmaA(x) \fto (\sigmaAp \circ \varphi_n^{-1})(x)\right).
\end{equation}
Now, let us choose some arbitrary $x \in A$ and focus on the expression $(\sigmaAp \circ \varphi_n^{-1})(x)$. It can be expressed as follows:
\[
(\sigmaAp \circ \varphi_n^{-1})(x) = \bigvee_{x' \in A'} \left(\sigmaAp(x') \fand \varphi_n(x, x')\right).
\]
Since~\eqref{eq: JHFLS 1} holds, we further get:
\begin{eqnarray*}
    (\sigmaAp \circ \varphi_n^{-1})(x) & \le & \bigvee_{x' \in A'} \left(\sigmaAp(x') \fand S(\bLdb(\mA_x), \bLdb(\mAp_{x'}))\right) \\
    & = & \bigvee_{x' \in A'} \left(\sigmaAp(x') \fand \bigwedge_{\substack{\ u \in \Sigma^{*} \\ |u| \le n}} (\bL(\mA_x)(u) \fto \bL(\mAp_{x'})(u))\right) \\
    & \le & \bigvee_{x' \in A'} \left(\sigmaAp(x') \fand (\bL(\mA_x)(u) \fto \bL(\mAp_{x'})(u))\right), 
\end{eqnarray*}
where $u \in \Sigma^{*}$ is an arbitrary word with $|u| \le n$ and the last inequality is justified by~\eqref{fop: GDJSK 10}. Now, applying~\eqref{fop: GDJSK 80a}, we obtain:
\[
(\sigmaAp \circ \varphi_n^{-1})(x) \le \bigvee_{x' \in A'} \left(\bL(\mA_x)(u) \fto (\sigmaAp(x') \fand \bL(\mAp_{x'})(u))\right).
\]
We further apply~\eqref{fop: GDJSK 230} and get the following:
\begin{eqnarray*}
(\sigmaAp \circ \varphi_n^{-1})(x) & \le & \left(\bL(\mA_x)(u) \fto \bigvee_{x' \in A'} (\sigmaAp(x') \fand \bL(\mAp_{x'})(u))\right) \\
& = & \left(\bL(\mA_x)(u) \fto \bL(\mAp)(u)\right).
\end{eqnarray*}
Now, we successively use~\eqref{fop: GDJSK 20}
and~\eqref{fop: GDJSK 100} to evaluate the expression $\sigmaA(x) \fto (\sigmaAp \circ \varphi_n^{-1})(x)$ in the following way:
\begin{eqnarray*}
\left(\sigmaA(x) \fto (\sigmaAp \circ \varphi_n^{-1})(x)\right) & \le & \left(\sigmaA(x) \fto (\bL(\mA_x)(u) \fto \bL(\mAp)(u))\right) \\
& \le & \left((\sigmaA(x) \fand \bL(\mA_x)(u)) \fto \bL(\mAp)(u)\right).
\end{eqnarray*}
By substituting the last inequality into~\eqref{eq.NormVarphinmAmAp} and by using~\eqref{fop: GDJSK 210}, we get:
\begin{eqnarray*}
\normS{\varphi_n}{\mA}{\mAp} & \le & \bigwedge_{x \in A} \left( (\sigmaA(x) \fand \bL(\mA_x)(u)) \fto \bL(\mAp)(u) \right)\\
& = & \left( \bigvee_{x \in A} (\sigmaA(x) \fand \bL(\mA_x)(u))  \right) \fto \bL(\mAp)(u) \\
& = & \left(\bL(\mA)(u) \fto \bL(\mAp)(u)\right).
\end{eqnarray*}
As the last inequality is valid for every word $u \in \Sigma^{*}$ with $|u| \le n$, we finally obtain:
\[
\normS{\varphi_n}{\mA}{\mAp} \le \bigwedge_{\substack{\ u \in \Sigma^{*} \\ |u| \le n}} \left(\bL(\mA)(u) \fto \bL(\mAp)(u)\right)
= S(\bLdb(\mA), \bLdb(\mAp)),
\]
and the proof is completed.
\myend
\end{proof}

The following corollary states that, if $\Phi$ is a depth-bounded fuzzy simulation between $\mA$ and $\mAp$, then the fuzzy degree in which the fuzzy language recognized by $\mA$ is a subset of the fuzzy language recognized by $\mAp$ is greater than or equal to $\nZdbs$. It follows immediately from~\eqref{eq: JHFLS 2}.   

\begin{corollary}\label{cor: JHDJA}
If $\Phi$ is a depth-bounded fuzzy simulation between fuzzy automata $\mA$ and $\mAp$, then 
\[ \nZdbs \leq S(\bL(\mA), \bL(\mAp)). \]
\end{corollary}

%\subsection{The Hennessy-Milner Property}

In the rest of this section, we present a result on the Hennessy-Milner property of depth-bounded fuzzy simulations between fuzzy automata. It gives a logical characterization of the greatest depth-bounded fuzzy simulation between two fuzzy automata.

The language $\mFs$ defined in~\cite{FB4FA} is the smallest set of formulas over~$\Sigma$ and~$\mL$ such that:
\begin{itemize}
	\item $\tau \in \mFs$; 
	\item if $s \in \Sigma$ and $w \in \mFs$, then $(s \circ \alpha) \in \mFs$;
	\item if $a \in L$ and $\alpha \in \mFs$, then $(a \to \alpha) \in \mFs$;
	\item if $\alpha, \beta \in \mFs$, then $(\alpha \land \beta) \in \mFs$.
\end{itemize}

For $n \in \NN$, let $\mFdbs{n}$ be the sublanguage of $\mFs$ that consists of formulas with the nesting depth of $\circ$ less than or equal to $n$. Formally, $(\mFdbs{n})_{n \in \NN}$ is the family of the smallest sets of formulas over~$\Sigma$ and~$\mL$ such that:
\begin{itemize}
\item $\tau \in \mFdbs{n}$; 
\item if $s \in \Sigma$ and $\alpha \in \mFdbs{n}$, then $(s \circ \alpha) \in \mFdbs{n+1}$;
\item if $a \in L$ and $\alpha \in \mFdbs{n}$, then $(a \to \alpha) \in \mFdbs{n}$;
\item if $\alpha, \beta \in \mFdbs{n}$, then $(\alpha \land \beta) \in \mFdbs{n}$.
\end{itemize}

We have 
\[ \mFs = \bigvee_{n \in \NN} \mFdbs{n}. \]

The meaning of formulas of $\mFs$ is explained in~\cite{FB4FA} as follows. 
Given a fuzzy automaton $\mA$ (over $\Sigma$ and $\mL$), a state $x \in A$ and a formula $\alpha \in \mFs$, the fuzzy degree in which $x$ has the property $\alpha$ is denoted by $\alpha^\mA(x)$. It is a value from $L$ with the following intuition:
\begin{itemize}
\item $\tau^\mA(x)$ is the degree in which $x$ is a terminal state;
\item $(s \circ \alpha)^\mA(x)$ is the degree in which executing the action $s$ at the state $x$ may lead to a state with the property $\alpha$;  
\item $(a \to \alpha)^\mA(x)$ is the degree in which $\alpha^\mA(x) \geq a$;
\item $(\alpha \land \beta)^\mA(x)$ is the degree in which $x$ has both the properties $\alpha$ and $\beta$. 
\end{itemize} 
Formally, the value $\alpha^\mA(x)$ for $\alpha \in \mFs \setminus \{\tau\}$ and $x \in A$ is defined inductively as follows~\cite{FB4FA}: 
\begin{eqnarray*}
(s \circ \alpha)^\mA(x) & = & (\deltaA_s \circ \alpha^\mA)(x) \\
(a \to \alpha)^\mA(x) & = & a \fto \alpha^\mA(x) \\
(\alpha \land \beta)^\mA(x) & = & \alpha^\mA(x) \land \beta^\mA(x).
\end{eqnarray*}
Thus, for $\alpha \in \mFs$, $\alpha^\mA$ is a fuzzy subset of~$A$.

\begin{example}\label{example: LJDNS 1b}
Let $\tuple{L, \leq}$, $\Sigma$, $\mA$ and $\mAp$ be as in Examples~\ref{example: LJDNS 0} and~\ref{example: LJDNS 1}. 
Consider the following formula 
\[ \alpha = (s \circ s \circ (0.9 \to \tau)) \] 
of $\mFdbs{2}$. We have 
\begin{eqnarray*}
\alpha^\mA(u) & = & 0.4 \fand 0.5 \fand (0.9 \fto 1) \\
\alpha^\mAp(u') & = &  0.5 \fand 0.4 \fand (0.9 \fto 0.8), 
\end{eqnarray*}
which give 
\begin{center}
\begin{NiceTabular}{@{}p{6em}p{6.5em}p{6.5em}p{6.5em}@{}}[]
    \toprule
     & \textbf{G\"odel} & \textbf{{\L}ukasiewicz} & \textbf{Product} \\
    \midrule
    $\alpha^\mA(u)$ & 0.4 & 0 & 0.2 \\
    $\alpha^\mAp(u')$ & 0.4 & 0 & 8/45 \\
    \bottomrule
\end{NiceTabular}
\end{center}
\myend
\end{example}

The following lemma is a generalization of the assertion~\eqref{eq: JHFLS 1} of Theorem~\ref{theorem: JHFLS}, as it implies that, if $(\varphi_n)_{n \in \NN}$ is a depth-bounded fuzzy simulation between fuzzy automata $\mA$ and $\mAp$, then for every $n \in \NN$ and every $\tuple{x,x'} \in A \times A'$:
\begin{equation}\label{eq: HDJHS}
	\varphi_n(x,x') \leq \bigwedge_{\alpha \in \mFdbs{n}}\! (\alpha^\mA(x) \fto \alpha^\mAp(x')).
\end{equation}
This inequality states that the formulas of $\mFdbs{n}$ are fuzzily preserved by $\varphi_n$. 

\begin{lemma}\label{lemma: KHSBS}
If $\Phi = (\varphi_n)_{n \in \NN}$ is a depth-bounded fuzzy simulation between fuzzy automata~$\mA$ and~$\mAp$, then for every $n \in \NN$ and every $\alpha \in \mFdbs{n}$: 
\[
	\varphi_n^{-1} \circ \alpha^\mA \leq \alpha^\mAp.
\]
\end{lemma}

The proof of this lemma is given in the appendix.

The following theorem is about the Hennessy-Milner property of depth-bounded fuzzy simulations between fuzzy automata. 
Its proof is given in the appendix.

\begin{theorem}\label{theorem: HDJHA}
Suppose that $\mL$ is linear and $\fand$ is continuous. 
Let $\mA$ and $\mAp$ be fuzzy automata, where $\mAp$ is image-finite. 
For $n \in \NN$, let $\varphi_n : A \times A' \to L$ be the fuzzy relation defined as follows:
\[
	\varphi_n(x,x') = \bigwedge_{\alpha \in \mFdbs{n}}\! (\alpha^\mA(x) \fto \alpha^\mAp(x')).
\]
%for $\tuple{x,x'} \in A \times A'$. 
Then, $(\varphi_n)_{n \in \NN}$ is the greatest depth-bounded fuzzy simulation between $\mA$ and $\mAp$. 
\end{theorem}

%===============================================================================

\section{Depth-Bounded Fuzzy Bisimulations between Fuzzy Automata}
\label{section: fbs}

In this section, we formally define depth-bounded fuzzy bisimulations between fuzzy automata and study their properties. In particular, we investigate their relationship with fuzzy bisimulations, the fuzzy invariance of the fuzzy length-bounded languages recognized by a fuzzy automaton under depth-bounded fuzzy bisimulations, as well as the Hennessy-Milner property of such bisimulations.

A~{\em depth-bounded fuzzy bisimulation} between fuzzy automata $\mA$ and $\mAp$ is a sequence $\Phi = (\varphi_n)_{n \in \NN}$ of fuzzy relations \mbox{$\varphi_n: A \times A' \to L$} such that:
\begin{itemize}
\item $\Phi$ is a depth-bounded fuzzy simulation between $\mA$ and $\mAp$, 
\item $\Phi^{-1}$ is a depth-bounded fuzzy simulation between $\mAp$ and $\mA$. 
\end{itemize}
That is, $(\varphi_n)_{n \in \NN}$ is a depth-bounded fuzzy bisimulation between $\mA$ and $\mAp$ if it satisfies the following conditions:
\begin{eqnarray}
\varphi_n & \leq & \varphi_{n-1}\quad \textrm{for all $n \geq 1$} \nonumber \\
\varphi_0^{-1} \circ \tauA & \leq & \tauAp \nonumber \\
\varphi_0 \circ \tauAp & \leq & \tauA \label{eq: HFKJA 4} \\
\varphi_n^{-1} \circ \deltaA_s & \leq & \deltaAp_s \circ \varphi_{n-1}^{-1}\quad \textrm{for all $s \in \Sigma$ and $n \geq 1$} \nonumber \\
\varphi_n \circ \deltaAp_s & \leq & \deltaA_s \circ \varphi_{n-1}\quad\, \textrm{for all $s \in \Sigma$ and $n \geq 1$}. \label{eq: HFKJA 5}
\end{eqnarray}

%A~{\em depth-bounded fuzzy auto-bisimulation} of a fuzzy automaton $\mA$ is a depth-bounded fuzzy bisimulation between $\mA$ and itself.
%
The {\em norm} of a depth-bounded fuzzy bisimulation $\Phi$ between $\mA$ and $\mAp$ is defined to be: 
\[ \nZdbbs \ =\ \nZdbs \land \normS{\Phi^{-1}}{\mAp}{\mA}. \]
Additionally, a depth-bounded fuzzy bisimulation between $\mA$ and itself is called a~{\em depth-bounded fuzzy auto-bisimulation} of~$\mA$.

\begin{example}\label{example: LJDNS 2}
Let $\tuple{L, \leq}$, $\Sigma$, $\mA$ and $\mAp$ be as in Example~\ref{example: LJDNS 0}. Recall that the fuzzy automata $\mA$ and $\mAp$ are depicted in Figure~\ref{fig: HDBKA} (on page~\pageref{fig: HDBKA}). 
For the cases where $\fand$ is the G\"odel, {\L}ukasiewicz or product t-norm, it can be checked that the sequence $\Phi = (\varphi_n)_{n \in \NN}$ specified below is the greatest depth-bounded fuzzy bisimulation between $\mA$ and~$\mAp$. 
\begin{center}
    \begin{NiceTabular}{@{}p{7em}p{25em}p{7em}@{}}[]
        \toprule
        \textbf{T-norm} & \textbf{Depth-bounded fuzzy bisimulation} & \textbf{Norm}\\
        \textbf{($\fand$)} & \textbf{($\Phi = (\varphi_n)_{n \in \NN}$)} & \textbf{($\nZdbs$)}\\
        \midrule
        G\"odel & $\varphi_0 = \{\tuple{u,u'}\!:\!1, \tuple{v,v'}\!:\!0.8\}$ & $\nZdbs = 0.4$ \\
                & $\varphi_n = \{\tuple{u,u'}\!:\!0.4, \tuple{v,v'}\!:\!0.4\}$ for all $n \geq 1$ &  \\
         \midrule
        {\L}ukasiewicz & $\varphi_0 = \{\tuple{u,u'}\!:\!1, \tuple{u,v'}\!:\!0.2, \tuple{v,v'}\!:\!0.8\}$ & $\nZdbs = 0.5$ \\
	             & $\varphi_1 = \{\tuple{u,u'}\!:\!0.7, \tuple{u,v'}\!:\!0.2, \tuple{v,v'}\!:\!0.7\}$ & \\
	           & $\varphi_2 = \{\tuple{u,u'}\!:\!0.6, \tuple{u,v'}\!:\!0.2, \tuple{v,v'}\!:\!0.6\}$ & \\
	           & $\varphi_n = \{\tuple{u,u'}\!:\!0.5, \tuple{u,v'}\!:\!0.2, \tuple{v,v'}\!:\!0.5\}$ for all $n \geq 3$ & 
 	 \\
         \midrule
        product & $\varphi_0 = \{\tuple{u,u'}\!:\!1, \tuple{v,v'}\!:\!0.8\}$ & $\nZdbs = 0$  \\
	        & $\varphi_n =\{\tuple{u,u'}\!:\!0.8^{n+1}, \tuple{v,v'}\!:\!0.8^{n+1}\}$ for all $n \geq 1$ & 
 	\\
        \bottomrule
    \end{NiceTabular}
\end{center}
\myend
\end{example}

The following proposition is a counterpart of Proposition~\ref{prop: JHFNW}.

\begin{proposition}\label{prop: JHFNW 2}
Let $\mA$, $\mAp$ and $\mAdp$ be fuzzy automata. Furthermore, let $\mathbf{\Phi}$ be a family of depth-bounded fuzzy bisimulations between $\mA$ and $\mAp$, and $\Psi$ a depth-bounded fuzzy bisimulation between $\mAp$ and $\mAdp$. Then, the following properties hold:
\begin{enumerate}[(a)]
\item If $\Phi_1, \Phi_2 \in \mathbf{\Phi}$ and $\Phi_1 \le \Phi_2$, then 
\[ \normBS{\Phi_1}{\mA}{\mAp} \leq \normBS{\Phi_2}{\mA}{\mAp}. \]
\item If $\Phi \in \mathbf{\Phi}$, then the inverse $\Phi^{-1}$ is a depth-bounded fuzzy bisimulation between $\mAp$ and $\mA$ with
\[ \normBS{\Phi^{-1}}{\mAp}{\mA} = \normBS{\Phi}{\mA}{\mAp}. \]  
\item If $\Phi \in \mathbf{\Phi}$, then the composition $\Phi \circ \Psi$ is a depth-bounded fuzzy bisimulation between $\mA$ and $\mAdp$ such that 
\[ \normBS{\Phi}{\mA}{\mAp} \fand \normBS{\Psi}{\mAp}{\mAdp} \leq \normBS{\Phi \circ \Psi}{\mA}{\mAdp}. \]
\item The meet $\bigvee\!\mathbf{\Phi}$ is also a depth-bounded fuzzy bisimulation between $\mA$ and $\mAp$. 
\item There always exists the greatest depth-bounded fuzzy bisimulation between $\mA$ and~$\mAp$.
\end{enumerate}
\end{proposition}
\begin{proof}
The assertion~(a) follows from~\eqref{fop: GDJSK 10} and~\eqref{fop: GDJSK 20}, while~(b) follows directly from the definition of depth-bounded fuzzy bisimulations. The assertion~(c) follows from the assertions~(b) of Proposition~\ref{prop: JHFNW} and the current proposition, while~(d) follows from the assertion~(c) of Proposition~\ref{prop: JHFNW} and the assertion~(b) of the current proposition, as well as from the fact that $(\bigvee\!\mathbf{\Phi})^{-1} = \bigvee\{\Phi^{-1} \mid \Phi \in \mathbf{\Phi}\}$. Finally, (e)~follows directly from~(d).
\myend
\end{proof}

The following theorem relates depth-bounded fuzzy bisimulations to fuzzy bisimulations. Roughly speaking, a depth-bounded fuzzy bisimulation is a decreasing sequence of fuzzy relations that, under certain light conditions, converges to a fuzzy bisimulation.

\begin{theorem}\label{theorem: HDKAK 2}
Suppose that $\mL$ satisfies the join-meet distributivity law and $\fand$ is continuous. Let $\Phi = (\varphi_n)_{n \in \NN}$ be a depth-bounded fuzzy bisimulation between image-finite fuzzy automata $\mA$ and $\mAp$. Then, the fuzzy relation $\varphi = \bigwedge\!\Phi$ is a fuzzy bisimulation between $\mA$ and $\mAp$ with $\nZbs = \nZdbbs$.
\end{theorem}

This theorem is a counterpart of Theorem~\ref{theorem: HDKAK}. It directly follows from Theorem~\ref{theorem: HDKAK} and the fact that $(\bigwedge\!\Phi)^{-1} =\bigwedge(\Phi^{-1})$.

The following corollary states that, under certain light conditions, the greatest depth-bounded fuzzy bisimulation between two fuzzy automata is a convergent sequence of fuzzy relations that approximate the greatest fuzzy bisimulation between the fuzzy automata. 

\begin{corollary}\label{cor: HGRKW 2}
Let $\Phi = (\varphi_n)_{n \in \NN}$ be the greatest depth-bounded fuzzy bisimulation between image-finite fuzzy automata $\mA$ and $\mAp$. If $\mL$ satisfies the join-meet distributivity law and $\fand$ is continuous, then $\bigwedge\!\Phi$ is the greatest fuzzy bisimulation between~$\mA$ and~$\mAp$.
\end{corollary}

This corollary is a counterpart of Corollary~\ref{cor: HGRKW}. It can be proved as done for Corollary~\ref{cor: HGRKW}, but using Theorem~\ref{theorem: HDKAK 2} instead of~\ref{theorem: HDKAK} and replacing the occurrences of ``simulation'' with ``bisimulation''.

The following corollary is a counterpart of Corollary~\ref{cor: HRKQA}.

\begin{corollary}\label{cor: HRKQA 2}
Let $\Phi = (\varphi_n)_{n \in \NN}$ be the greatest depth-bounded fuzzy auto-bisimulation of $\mA$. If $\mL$ satisfies the join-meet distributivity law, $\fand$ is continuous and $\mA$ is image-finite, then 
\begin{enumerate}[(a)]
\item $\bigwedge\!\Phi$ is the greatest fuzzy auto-bisimulation of $\mA$,
\item each $\varphi_n$ is a fuzzy equivalence,
\item $\normBS{\Phi}{\mA}{\mA} = 1$.
\end{enumerate}
\end{corollary}

\begin{proof}
Part~(a) follows from Corollary~\ref{cor: HGRKW 2}. To prove~(b), observe that the sequence $\Psi$ that consists of infinite many $id_A$ is a depth-bounded fuzzy auto-bisimulation of $\mA$. Hence, each $\varphi_n$ is reflexive. By Proposition~\ref{prop: JHFNW 2}, $\Phi \circ \Phi \leq \Phi$. Hence, each $\varphi_n$ is transitive. 
Also by Proposition~\ref{prop: JHFNW 2}, $\Phi^{-1}$ is a depth-bounded fuzzy auto-bisimulation of $\mA$. Hence, $\Phi^{-1} \leq \Phi$. Consequently, $\Phi \leq \Phi^{-1}$ and then $\Phi = \Phi^{-1}$. Therefore, each $\varphi_n$ is a fuzzy equivalence. Finally, to prove (c), observe that $\normBS{\Psi}{\mA}{\mA} = 1$. By Proposition~\ref{prop: JHFNW 2}, $\normBS{\Psi}{\mA}{\mA} \leq \normBS{\Phi}{\mA}{\mA}$. Hence, $\normBS{\Phi}{\mA}{\mA} = 1$.
\myend
\end{proof}

%\subsection{Invariance of the Recognized Language}

The following theorem is a counterpart of Theorem~\ref{theorem: JHFLS}. It states a kind of fuzzy invariance of the fuzzy length-bounded languages recognized by a fuzzy automaton under depth-bounded fuzzy bisimulations. If $(\varphi_n)_{n \in \NN}$ is a depth-bounded fuzzy bisimulation between $\mA$ and $\mAp$, then for $\tuple{x,x'} \in A \times A'$, the fuzzy degree in which $\bLdb(\mA_x)$ and $\bLdb(\mAp_{x'})$ are equivalent is greater than or equal to $\varphi_n(x,x')$ and, furthermore, the fuzzy degree in which $\bLdb(\mA)$ and $\bLdb(\mAp)$ are equivalent is greater than or equal to $\normBS{\varphi_n}{\mA}{\mAp}$.

\begin{theorem}\label{theorem: HGFKW}
Let $\Phi = (\varphi_n)_{n \in \NN}$ be a depth-bounded fuzzy bisimulation between fuzzy automata~$\mA$ and~$\mAp$. Then, for every $n \in \NN$ and every $\tuple{x,x'} \in A \times A'$:
\begin{eqnarray}
\varphi_n(x,x') & \leq & E(\bLdb(\mA_x), \bLdb(\mAp_{x'})), \label{eq: HGFKW 1} \\[0.5ex]
\normBS{\varphi_n}{\mA}{\mAp} & \leq & E(\bLdb(\mA), \bLdb(\mAp)). \label{eq: HGFKW 2} 
\end{eqnarray}
\end{theorem}

The assertion~\eqref{eq: HGFKW 1} follows from the assertion~\eqref{eq: JHFLS 1} of Theorem~\ref{theorem: JHFLS} and the second assertion of Proposition~\eqref{prop: JHFNW 2}. 
The assertion~\eqref{eq: HGFKW 2} follows from the assertion~\eqref{eq: JHFLS 2} of Theorem~\ref{theorem: JHFLS}. 

The following corollary is a counterpart of Corollary~\ref{cor: JHDJA}. It states that, if $\Phi$ is a depth-bounded fuzzy bisimulation between~$\mA$ and~$\mAp$, then the fuzzy degree in which the fuzzy languages recognized by~$\mA$ and~$\mAp$ are equivalent is greater than or equal to $\nZdbbs$. It follows immediately from~\eqref{eq: HGFKW 2}.   

\begin{corollary}\label{cor: JHDJA 2}
If $\Phi$ is a depth-bounded fuzzy bisimulation between fuzzy automata $\mA$ and $\mAp$, then 
\[ \nZdbbs \leq E(\bL(\mA), \bL(\mAp)). \]
\end{corollary}

%\subsection{The Hennessy-Milner Property}

In the rest of this section, we present a result on the Hennessy-Milner property of depth-bounded fuzzy bisimulations between fuzzy automata. It gives a logical characterization of the greatest depth-bounded fuzzy bisimulation between two fuzzy automata.

The set $\mFbs$ of formulas is defined in~\cite{FB4FA} in a similar way as for $\mFs$, but using expressions of the form $(a \leftrightarrow \alpha)$ instead of $(a \to \alpha)$. 
For $n \in \NN$, let $\mFdbbs{n}$ be the sublanguage of $\mFbs$ that consists of formulas with the nesting depth of $\circ$ less than or equal to $n$. Formally, $(\mFdbbs{n})_{n \in \NN}$ is the family of the smallest sets of formulas over~$\Sigma$ and~$\mL$ such that:
\begin{itemize}
	\item $\tau \in \mFdbbs{n}$; 
	\item if $s \in \Sigma$ and $\alpha \in \mFdbbs{n}$, then $(s \circ \alpha) \in \mFdbbs{n+1}$;
	\item if $a \in L$ and $\alpha \in \mFdbbs{n}$, then $(a \leftrightarrow \alpha) \in \mFdbbs{n}$;
	\item if $\alpha, \beta \in \mFdbbs{n}$, then $(\alpha \land \beta) \in \mFdbbs{n}$.
\end{itemize}

We have 
\[ \mFbs = \bigvee_{n \in \NN} \mFdbbs{n}. \]

The value $\alpha^\mA(x)$ for $\alpha \in \mFbs \setminus \{\tau\}$ and $x \in A$ is defined analogously as for the case where $\alpha \in \mFs \setminus \{\tau\}$, except that $(a \leftrightarrow \alpha)^\mA(x)$ is defined to be $(a \fequiv \alpha^\mA(x))$, i.e., the fuzzy degree in which $\alpha^\mA(x)$ is equal to~$a$.  

The following lemma is a counterpart of Lemma~\ref{lemma: KHSBS} and a generalization of the assertion~\eqref{eq: HGFKW 1} of Theorem~\ref{theorem: HGFKW}, as it implies that, if $(\varphi_n)_{n \in \NN}$ is a depth-bounded fuzzy bisimulation between fuzzy automata~$\mA$ and~$\mAp$, then for every $n \in \NN$ and $\tuple{x,x'} \in A \times A'$:
\begin{equation}\label{eq: HDJHS 2}
\varphi_n(x,x') \leq \bigwedge_{\alpha \in \mFdbbs{n}}\! (\alpha^\mA(x) \fequiv \alpha^\mAp(x')).
\end{equation}
This inequality states that the formulas of $\mFdbbs{n}$ are fuzzily invariant under $\varphi_n$. 

\begin{lemma}\label{lemma: KHSBS 2}
If $\Phi = (\varphi_n)_{n \in \NN}$ is a depth-bounded fuzzy bisimulation between fuzzy automata~$\mA$ and~$\mAp$, then for every $n \in \NN$ and every $\alpha \in \mFdbbs{n}$: 
\begin{eqnarray}
\varphi_n^{-1} \circ \alpha^\mA & \leq & \alpha^\mAp \label{eq: KHSBS 2a} \\
\varphi_n \circ \alpha^\mAp & \leq & \alpha^\mA. \label{eq: KHSBS 2b}
\end{eqnarray}
\end{lemma}

The proof of this lemma is given in the appendix.

The following theorem is about the Hennessy-Milner property of depth-bounded fuzzy bisimulations between fuzzy automata. 
Its proof is given in the appendix.

\begin{theorem}\label{theorem: HDJHA 2}
Suppose that $\mL$ is linear and $\fand$ is continuous. Let $\mA$ and $\mAp$ be image-finite fuzzy automata. For $n \in \NN$, let \mbox{$\varphi_n : A \times A' \to L$} be the fuzzy relation defined as follows:
\[
	\varphi_n(x,x') = \bigwedge_{\alpha \in \mFdbbs{n}}\! (\alpha^\mA(x) \fequiv \alpha^\mAp(x')).
\]
Then, $(\varphi_n)_{n \in \NN}$ is the greatest depth-bounded fuzzy bisimulation between $\mA$ and $\mAp$. 
\end{theorem}

%===============================================================================

\section{Computing Depth-Bounded Fuzzy Simulations and Bisimulations}
\label{section: computation}

In this section, let $\mA = \tuple{A, \deltaA, \sigmaA, \tauA}$ and $\mAp = \tuple{A', \deltaAp, \sigmaAp, \tauAp}$ be finite fuzzy automata over a finite alphabet $\Sigma$ and let $k$ be a natural number. We present algorithms that, given such inputs, compute the component $\varphi_k$ of the greatest depth-bounded fuzzy simulation (respectively, bisimulation) $(\varphi_i)_{i \in \NN}$ between $\mA$ and $\mAp$. 
We denote \mbox{$n = |A| + |A'|$} and \mbox{$m = |\deltaA| + |\deltaAp|$}, where $|\deltaA|$ is the number of tuples $\tuple{x,s,y} \in A \times \Sigma \times A$ such that $\deltaA_s(x,y) > 0$ and $|\deltaAp|$ has a similar meaning. 
Our algorithms have a time complexity of order $O(k(m+n)n)$, under the assumption that $|\Sigma|$ is a constant and the fuzzy operations $\land$, $\lor$, $\fand$ and $\fto$ can be done in constant time.

We assume that the representation of $\Sigma$, $\mA$ and $\mAp$ satisfies the following conditions:
\begin{itemize}
\item $\Sigma = 0..(|\Sigma| - 1)$, $A = 0..(|A|-1)$, $A' = 0..(|A'|-1)$, where $0..h$ denotes the set $\{0,1,\ldots,h\}$; 
\item $\sigmaA$, $\tauA$, $\sigmaAp$ and $\tauAp$ are represented as arrays; 
\item $\deltaA$ is represented as the array $\Succ$ of dimensions $|\Sigma| \times |A|$ such that, for $s \in \Sigma$ and $x \in A$, $\Succ[s][x]$ is a list of all pairs $(y,d)$ with $y \in A$ and $d \in L$ such that $d = \deltaA_s(x,y) > 0$;
\item $\deltaAp$ is represented as the array $\Succ'$ of dimensions $|\Sigma'| \times |A'|$ in a similar way.  
\end{itemize}

As a preprocessing, we construct 
\begin{itemize}
\item the array $\Pred$ of dimensions $|\Sigma| \times |A|$ such that, for $s \in \Sigma$ and $y \in A$, $\Pred[s][y]$ is a list of all pairs $(x,d)$ with $x \in A$ and $d \in L$ such that $d = \deltaA_s(x,y) > 0$;
\item the array $\Pred'$ of dimensions $|\Sigma| \times |A'|$ such that, for $s \in \Sigma$ and $y' \in A'$, $\Pred'[s][y']$ is a list of all pairs $(x',d)$ with $x' \in A'$ and $d \in L$ such that $d = \deltaAp_s(x',y') > 0$.
\end{itemize}

The names $\Succ$ and $\Pred$ stand for ``successor'' and ``predecessor'', respectively. 
The construction of the arrays $\Pred$ and $\Pred'$ can be done in time $O(m+n)$. 
We will write $\Succ_s$, $\Pred_s$, $\Succ'_s$ and $\Pred'_s$ to denote $\Succ[s]$, $\Pred[s]$, $\Succ'[s]$ and $\Pred'[s]$, respectively. 

\newcommand{\Bound}{\mathit{bound}}

\begin{algorithm}[t]
\caption{\CompDBFS\label{algCompDBFS}}
\Input{finite fuzzy automata $\mA = \tuple{A, \deltaA, \sigmaA, \tauA}$ and $\mAp = \tuple{A', \deltaAp, \sigmaAp, \tauAp}$ over a finite alphabet $\Sigma$, where $\deltaA$ and $\deltaAp$ are represented by the arrays $\Succ$ and $\Succ'$ specified in this section, together with a natural number $k$.}
\Output{the component $\varphi_k$ of the greatest depth-bounded fuzzy simulation $(\varphi_i)_{i \in \NN}$ between $\mA$ and $\mAp$.}
\LocalVariables{arrays $\varphi, \psi: A \times A' \to L$.}

\BlankLine

construct the array $\Pred$ specified in this section\label{step: algCompDBFS 1}\;
\ForEach{\em $x \in A$ and $x' \in A'$\label{step: algCompDBFS 2}}{
	$\varphi[x][x'] := (\tauA[x] \fto \tauAp[x'])$\label{step: algCompDBFS 3}\;
}

\ForEach{\em $i$ from $1$ to $k$\label{step: algCompDBFS 4}}{
	$\changed := \False$\;
	make $\psi$ a (deep) copy of $\varphi$\label{step: algCompDBFS 6}\;
	
	\ForEach{\em $s \in \Sigma$, $x' \in A'$ and $y \in A$\label{step: algCompDBFS 7}}{
		$\Bound := 0$\label{step: algCompDBFS 8}\;
		\ForEach{$(y',d) \in \Succ'_s[x']$}{
			$\Bound := \Bound \lor (d \fand \psi[y][y'])$\label{step: algCompDBFS 10}\;
		}
		\ForEach{$(x,d) \in \Pred_s[y]$\label{step: algCompDBFS 11}}{
			\If{$\varphi[x][x'] \not\leq (d \fto \Bound)$}{
				$\varphi[x][x'] := \varphi[x][x'] \land (d \fto \Bound)$\;
				$\changed := \True$\label{step: algCompDBFS 14}\;
			}
		}
	}

	\lIf{not $\changed$\label{step: algCompDBFS 15}}{\Break}
}

\Return $\varphi$\;
\end{algorithm}

\begin{algorithm}[H]
	\caption{\CompDBFB\label{algCompDBFB}}
	\Input{finite fuzzy automata $\mA = \tuple{A, \deltaA, \sigmaA, \tauA}$ and $\mAp = \tuple{A', \deltaAp, \sigmaAp, \tauAp}$ over a finite alphabet $\Sigma$, where $\deltaA$ and $\deltaAp$ are represented by the arrays $\Succ$ and $\Succ'$ specified in this section, together with a natural number $k$.}
	\Output{the component $\varphi_k$ of the greatest depth-bounded fuzzy bisimulation $(\varphi_i)_{i \in \NN}$ between $\mA$ and $\mAp$.}
	\LocalVariables{arrays $\varphi, \psi: A \times A' \to L$.}
	
	\BlankLine
	
	construct the arrays $\Pred$ and $\Pred'$ specified in this section\label{step: algCompDBFB 1}\;
	
	\ForEach{\em $x \in A$ and $x' \in A'$\label{step: algCompDBFB 2}}{
		$\varphi[x][x'] := (\tauA[x] \fequiv \tauAp[x'])$\label{step: algCompDBFB 3}\;
	}
	
	\ForEach{\em $i$ from $1$ to $k$\label{step: algCompDBFB 4}}{
		$\changed := \False$\;
		make $\psi$ a (deep) copy of $\varphi$\label{step: algCompDBFB 6}\;
		
		\ForEach{\em $s \in \Sigma$, $x' \in A'$ and $y \in A$\label{step: algCompDBFB 7}}{
			$\Bound := 0$\label{step: algCompDBFB 8}\;
			\ForEach{$(y',d) \in \Succ'_s[x']$}{
				$\Bound := \Bound \lor (d \fand \psi[y][y'])$\label{step: algCompDBFB 10}\;
			}
			\ForEach{$(x,d) \in \Pred_s[y]$\label{step: algCompDBFB 11}}{
				\If{$\varphi[x][x'] \not\leq (d \fto \Bound)$}{
					$\varphi[x][x'] := \varphi[x][x'] \land (d \fto \Bound)$\;
					$\changed := \True$\label{step: algCompDBFB 14}\;
				}
			}
		}
		
		\ForEach{\em $s \in \Sigma$, $x \in A$ and $y' \in A$\label{step: algCompDBFB 15}}{
			$\Bound := 0$\label{step: algCompDBFB 16}\;
			\ForEach{$(y,d) \in \Succ_s[x]$}{
				$\Bound := \Bound \lor (d \fand \psi[y][y'])$\label{step: algCompDBFB 18}\;
			}
			\ForEach{$(x',d) \in \Pred'_s[y']$\label{step: algCompDBFB 19}}{
				\If{$\varphi[x][x'] \not\leq (d \fto \Bound)$}{
					$\varphi[x][x'] := \varphi[x][x'] \land (d \fto \Bound)$\;
					$\changed := \True$\label{step: algCompDBFB 22}\;
				}
			}
		}
		
		\lIf{not $\changed$\label{step: algCompDBFB 23}}{\Break}
	}
	
	\Return $\varphi$\;
\end{algorithm}

Algorithm~\ref{algCompDBFS} (on page~\pageref{algCompDBFS}) is our algorithm for computing the component $\varphi_k$ of the greatest depth-bounded fuzzy simulation $(\varphi_i)_{i \in \NN}$ between $\mA$ and $\mAp$. It first constructs the array $\Pred$ specified above. Then, by the statements \ref{step: algCompDBFS 2} and \ref{step: algCompDBFS 3}, it initially sets $\varphi$ to $\varphi_0$, which is the greatest fuzzy relation satisfying~\eqref{eq: HFKJA 2}. 
Consider the $i$-th iteration of the main loop in the statements \ref{step: algCompDBFS 4}-\ref{step: algCompDBFS 15} of the algorithm. Its aim is to update $\varphi$ from $\varphi_{i-1}$ to $\varphi_i$. That is, the invariant of the loop (which holds before each iteration) is $\varphi = \varphi_{i-1}$. First, a (deep) copy of $\varphi$ is created and stored in $\psi$, with the intention to keep $\varphi_{i-1}$ without modifications during the iteration of the loop. To guarantee the condition~\eqref{eq: HFKJA 3} with~$n$ replaced by~$i$, for each $s \in \Sigma$, $x' \in A'$ and $y \in A$, the statements~\ref{step: algCompDBFS 8}-\ref{step: algCompDBFS 10} set $\Bound$ to $(\deltaAp_s \circ \varphi_{i-1}^{-1})(x',y)$, then the statements~\ref{step: algCompDBFS 11}-\ref{step: algCompDBFS 14} minimally reduce $\varphi$ so that $(\varphi^{-1} \circ \deltaA_s)(x',y) \leq \Bound$. That is, starting from $\varphi = \varphi_{i-1}$, the statements \ref{step: algCompDBFS 7}-\ref{step: algCompDBFS 14} minimally reduce $\varphi$ so that $\varphi^{-1} \circ \deltaA_s \leq \deltaAp_s \circ \varphi_{i-1}^{-1}$. By~\eqref{eq: HFKJA 1} and~\eqref{eq: HFKJA 3} with $n$ replaced by $i$, this must result in $\varphi = \varphi_i$ and the mentioned invariant of the loop holds. The main loop (in the statements \ref{step: algCompDBFS 4}-\ref{step: algCompDBFS 15}) uses the flag $\changed$, which is turned off before each iteration and turned on when $\varphi$ is reduced. So, if $\changed = \False$ at the end of an iteration, then the greatest fixpoint for $\varphi$ has been reached and the loop can be terminated. We conclude that Algorithm~\ref{algCompDBFS} is correct. This directly follows from the justified invariant of the main loop and the use of the flag $\changed$. 

We now analyze the complexity of Algorithm~\ref{algCompDBFS}:
\begin{itemize}
\item The statement~\ref{step: algCompDBFS 1} runs in time $O(m+n)$.
\item The loop in the statements~\ref{step: algCompDBFS 2} and~\ref{step: algCompDBFS 3} runs in time $O(n^2)$.
\item The statement~\ref{step: algCompDBFS 6} runs in time $O(n^2)$.
\item The loop in the statements~\ref{step: algCompDBFS 7}-\ref{step: algCompDBFS 14} runs in time $O((m+n)n)$. Recall that $|\Sigma|$ is a constant. 
\item The loop in the statements~\ref{step: algCompDBFS 4}-\ref{step: algCompDBFS 14} runs in time $O(k(m+n)n)$. 
\end{itemize} 
Totally, Algorithm~\ref{algCompDBFS} runs in time $O(k(m+n)n)$. 
We arrive at the following result.

\begin{theorem}
Algorithm~\ref{algCompDBFS} is correct. That is, given finite fuzzy automata $\mA$ and $\mAp$ over a finite alphabet $\Sigma$ together with a natural number $k$, it returns the component $\varphi_k$ of the greatest depth-bounded fuzzy simulation $(\varphi_i)_{i \in \NN}$ between $\mA$ and $\mAp$. It runs in time $O(k(m+n)n)$, where $n = |A| + |A'|$ and $m = |\deltaA| + |\deltaAp|$, under the assumption that $|\Sigma|$ is a constant and the fuzzy operations $\land$, $\lor$, $\fand$ and $\fto$ can be done in constant time. 
\end{theorem}

Algorithm~\ref{algCompDBFB} (on page~\pageref{algCompDBFB}) is our algorithm for computing the component $\varphi_k$ of the greatest depth-bounded fuzzy bisimulation $(\varphi_i)_{i \in \NN}$ between $\mA$ and $\mAp$. It can be explained analogously as done for Algorithm~\ref{algCompDBFS} before. The algorithm first constructs the arrays $\Pred$ and $\Pred'$ specified at the beginning of this section. Then, by the statements \ref{step: algCompDBFB 2} and \ref{step: algCompDBFB 3}, it initially sets $\varphi$ to $\varphi_0$, which is the greatest fuzzy relation satisfying~\eqref{eq: HFKJA 2} and~\eqref{eq: HFKJA 4}. 
Consider the $i$-th iteration of the main loop in the statements \ref{step: algCompDBFB 4}-\ref{step: algCompDBFB 23} of the algorithm. Its aim is to update $\varphi$ from $\varphi_{i-1}$ to $\varphi_i$. That is, the invariant of the loop (which holds before each iteration) is $\varphi = \varphi_{i-1}$. First, a (deep) copy of $\varphi$ is created and stored in $\psi$, with the intention to keep $\varphi_{i-1}$ without modifications during the iteration of the loop. 
To guarantee the condition~\eqref{eq: HFKJA 3} with $n$ replaced by $i$, for each $s \in \Sigma$, $x' \in A'$ and $y \in A$, the statements~\ref{step: algCompDBFB 8}-\ref{step: algCompDBFB 10} set $\Bound$ to $(\deltaAp_s \circ \varphi_{i-1}^{-1})(x',y)$, then the statements~\ref{step: algCompDBFB 11}-\ref{step: algCompDBFB 14} minimally reduce $\varphi$ so that $(\varphi^{-1} \circ \deltaA_s)(x',y) \leq \Bound$. 
Similarly, to guarantee the condition~\eqref{eq: HFKJA 5} with $n$ replaced by $i$, for each $s \in \Sigma$, $x \in A$ and $y' \in A'$, the statements~\ref{step: algCompDBFB 16}-\ref{step: algCompDBFB 18} set $\Bound$ to $(\deltaA_s \circ \varphi_{i-1})(x,y')$, then the statements~\ref{step: algCompDBFB 19}-\ref{step: algCompDBFB 22} minimally reduce $\varphi$ so that $(\varphi \circ \deltaAp_s)(x,y') \leq \Bound$. 
That is, starting from $\varphi = \varphi_{i-1}$, the statements \ref{step: algCompDBFB 7}-\ref{step: algCompDBFB 22} minimally reduce $\varphi$ so that 
$\varphi^{-1} \circ \deltaA_s \leq \deltaAp_s \circ \varphi_{i-1}^{-1}$ and
$\varphi \circ \deltaAp_s \leq \deltaA_s \circ \varphi_{i-1}$. 
By~\eqref{eq: HFKJA 1}, \eqref{eq: HFKJA 3} and~\eqref{eq: HFKJA 5} with $n$ replaced by $i$, this must result in $\varphi = \varphi_i$ and the mentioned invariant of the loop holds. The main loop (in the statements \ref{step: algCompDBFB 4}-\ref{step: algCompDBFB 23}) uses the flag $\changed$, which is turned off before each iteration and turned on when $\varphi$ is reduced. So, if $\changed = \False$ at the end of an iteration, then the greatest fixpoint for $\varphi$ has been reached and the loop can be terminated. We conclude that Algorithm~\ref{algCompDBFB} is correct. This directly follows from the justified invariant of the main loop and the use of the flag $\changed$. 
The complexity analysis for Algorithm~\ref{algCompDBFB} is similar to the one for Algorithm~\ref{algCompDBFS}. Thus, we arrive at the following result.

\begin{theorem}
Algorithm~\ref{algCompDBFB} is correct. That is, given finite fuzzy automata $\mA$ and $\mAp$ over a finite alphabet $\Sigma$ together with a natural number $k$, it returns the component $\varphi_k$ of the greatest depth-bounded fuzzy bisimulation $(\varphi_i)_{i \in \NN}$ between $\mA$ and $\mAp$. It runs in time $O(k(m+n)n)$, where $n = |A| + |A'|$ and $m = |\deltaA| + |\deltaAp|$, under the assumption that $|\Sigma|$ is a constant and the fuzzy operations $\land$, $\lor$, $\fand$ and $\fto$ can be done in constant time. 
\end{theorem}

We have implemented Algorithms~\ref{algCompDBFS} and~\ref{algCompDBFB} in Python and made the program publicly available~\cite{BFBA-prog}. The program uses the unit interval $[0,1]$ with the usual order as the lattice. The user can experiment with it using a built-in or user-defined t-norm together with its corresponding residuum. 

%===============================================================================

\section{Related Work}
\label{section: related work}

Earlier works on simulations or bisimulations for fuzzy systems include the works~\cite{CaoCK11,CaoSWC13} on crisp bisimulations between fuzzy transition systems (FTSs), the works~\cite{CiricIDB12,CiricIJD12,Jancic14} on simulations and bisimulations between fuzzy automata, the works~\cite{EleftheriouKN12,Fan15} on crisp/fuzzy bisimulations for modal logics and the works~\cite{ai/FanL14,IgnjatovicCS15} on crisp/fuzzy bisimulations for weighted/fuzzy social networks. Simulations and bisimulations introduced and studied in~\cite{CiricIDB12,CiricIJD12,Jancic14} are fuzzy relations, but they compare fuzzy automata in a crisp manner. 

Concerning logical characterizations of simulations and bisimulations, there are the works~\cite{DBLP:journals/ijar/PanC0C14,DBLP:journals/ijar/PanLC15,fuzzIEEE/NguyenN21,ijar/Nguyen21,DBLP:journals/fss/WuD16} on logical characterizations of crisp or fuzzy simulations between FTSs, the works~\cite{DBLP:journals/ijar/WuCHC18,DBLP:journals/fss/WuCBD18,DBLP:journals/fss/WuD16} on logical characterizations of crisp bisimulations between FTSs, the works~\cite{Fan15,aml/MartiM18,fuin/Diaconescu20,NguyenFSS2021} on logical characterizations of crisp or fuzzy bisimulations for modal logics, the works~\cite{FSS2020,tfs/NguyenN23,Nguyen-TFS2019} on logical characterizations of crisp or fuzzy bisimulations/bisimilarity for fuzzy description logics, and the work~\cite{ai/FanL14} on logical characterizations of crisp/fuzzy bisimulations for weighted social networks. 

Concerning computation of simulations and bisimulations, there are the works~\cite{CompCB-arxiv,StanimirovicSC2019,DBLP:journals/fss/WuCBD18} on computing crisp bisimulations for fuzzy structures (FTSs, fuzzy automata or fuzzy labeled graphs), the works~\cite{CiricIJD12,IgnjatovicCS15,MNS.22,isci/Nguyen23,FuzzyMinimaxNets,TFS2020} on computing fuzzy bisimulations for fuzzy structures (fuzzy automata, fuzzy social networks, fuzzy labeled graphs or fuzzy interpretations in fuzzy description logics), the works~\cite{CiricIJD12,FuzzyMinimaxNets,TFS2020} on computing fuzzy simulations for fuzzy structures (fuzzy automata, fuzzy labeled graphs or fuzzy interpretations in fuzzy description logics), and the work~\cite{DBLP:journals/jifs/Nguyen22} on computing crisp simulations between FTSs. 
The currently known algorithms for computing the greatest fuzzy simulation or bisimulation between two finite fuzzy labeled graphs (which can represents various kinds of fuzzy structures, including fuzzy automata) have an exponential time complexity when the {\L}ukasiewicz or product structure of fuzzy values is used~\cite{FuzzyMinimaxNets}. The procedure given in~\cite{IgnjatovicCS15} for computing the greatest fuzzy bisimulation between two fuzzy social networks runs in time $O(ln^5)$, where $n$ is the size of the networks and $l$ is the number of different fuzzy values generated by the procedure. This value $l$ may be infinite when the product structure is used, or arbitrarily big (independently from the input size) when the {\L}ukasiewicz structure is used. 

Other notable related works include~\cite{DuZ18a,jcss/CIRIC2010,MicicJS18,DivroodiHNN18,QiaoZ21,fss/QIAO2023,tfs/QIAO2023}. In~\cite{DivroodiHNN18} depth-bounded bisimulations between (crisp) interpretations in description logics are used for concept learning. 
The work~\cite{QiaoZ21} defines limited approximate simulations and bisimulations for the so-called quantitative fuzzy approximation spaces (QFASs). They act as crisp binary relations on the set of states of a QFAS. The works \cite{fss/QIAO2023,tfs/QIAO2023} introduce the notions of bisimulations for FTSs that are based on a fuzzy relational lifting method via fuzzy similarity measures induced by residua (implications) in complete residuated lattices.

%===============================================================================

\section{Conclusions}
\label{section: cons}

We have introduced the notions of depth-bounded fuzzy simulations/bisimulations between fuzzy automata. We have proved that, under some light conditions, they give approximations of fuzzy simulations/bisimulations (Theorems~\ref{theorem: HDKAK} and~\ref{theorem: HDKAK 2}). These approximations are good in that, on one hand, Theorem~\ref{theorem: JHFLS} (respectively, \ref{theorem: HGFKW}) states that the fuzzy length-bounded languages recognized by a fuzzy automaton are fuzzily preserved by depth-bounded fuzzy simulations (respectively, fuzzily invariant under depth-bounded fuzzy bisimulations), while on the other hand, they overcome the unexpected phenomenon of fuzzy simulations/bisimulations discussed in the introduction section. We have provided a logical characterization of the greatest depth-bounded fuzzy simulation or bisimulation between two fuzzy automata (Theorems~\ref{theorem: HDJHA} and~\ref{theorem: HDJHA 2}). We have also given polynomial-time algorithms for computing the $n$th component of the greatest depth-bounded fuzzy simulation (respectively, bisimulation) between two finite fuzzy automata. These algorithms are of a particular importance in the context that the currently known algorithms for computing the greatest fuzzy (bi)simulation between two finite fuzzy automata have an exponential time complexity when the {\L}ukasiewicz or product structure of fuzzy values is used~\cite{FuzzyMinimaxNets}.

%===============================================================================

\bibliography{BSfDL}

\begin{thebibliography}{44}
\expandafter\ifx\csname natexlab\endcsname\relax\def\natexlab#1{#1}\fi
\providecommand{\url}[1]{\texttt{#1}}
\providecommand{\href}[2]{#2}
\providecommand{\path}[1]{#1}
\providecommand{\DOIprefix}{doi:}
\providecommand{\ArXivprefix}{arXiv:}
\providecommand{\URLprefix}{URL: }
\providecommand{\Pubmedprefix}{pmid:}
\providecommand{\doi}[1]{\href{http://dx.doi.org/#1}{\path{#1}}}
\providecommand{\Pubmed}[1]{\href{pmid:#1}{\path{#1}}}
\providecommand{\bibinfo}[2]{#2}
\ifx\xfnm\relax \def\xfnm[#1]{\unskip,\space#1}\fi
%Type = Book
\bibitem[{B\v{e}lohl{\'a}vek(2002)}]{Belohlavek2002}
\bibinfo{author}{B\v{e}lohl{\'a}vek, R.}, \bibinfo{year}{2002}.
\newblock \bibinfo{title}{Fuzzy Relational Systems: Foundations and
  Principles}.
\newblock \bibinfo{publisher}{Kluwer}.
%Type = Article
\bibitem[{Cao et~al.(2011)Cao, Chen and Kerre}]{CaoCK11}
\bibinfo{author}{Cao, Y.}, \bibinfo{author}{Chen, G.}, \bibinfo{author}{Kerre,
  E.}, \bibinfo{year}{2011}.
\newblock \bibinfo{title}{Bisimulations for fuzzy-transition systems}.
\newblock \bibinfo{journal}{{IEEE} Trans. Fuzzy Systems} \bibinfo{volume}{19},
  \bibinfo{pages}{540--552}.
\newblock \DOIprefix\doi{10.1109/TFUZZ.2011.2117431}.
%Type = Article
\bibitem[{Cao et~al.(2013)Cao, Sun, Wang and Chen}]{CaoSWC13}
\bibinfo{author}{Cao, Y.}, \bibinfo{author}{Sun, S.}, \bibinfo{author}{Wang,
  H.}, \bibinfo{author}{Chen, G.}, \bibinfo{year}{2013}.
\newblock \bibinfo{title}{A behavioral distance for fuzzy-transition systems}.
\newblock \bibinfo{journal}{{IEEE} Trans. Fuzzy Systems} \bibinfo{volume}{21},
  \bibinfo{pages}{735--747}.
\newblock \DOIprefix\doi{10.1109/TFUZZ.2012.2230177}.
%Type = Article
\bibitem[{{\'C}iri{\'c} et~al.(2012a){\'C}iri{\'c}, Ignjatovi{\'c},
  Damljanovi{\'c} and Ba\v{s}i{\'c}}]{CiricIDB12}
\bibinfo{author}{{\'C}iri{\'c}, M.}, \bibinfo{author}{Ignjatovi{\'c}, J.},
  \bibinfo{author}{Damljanovi{\'c}, N.}, \bibinfo{author}{Ba\v{s}i{\'c}, M.},
  \bibinfo{year}{2012}a.
\newblock \bibinfo{title}{Bisimulations for fuzzy automata}.
\newblock \bibinfo{journal}{Fuzzy Sets and Systems} \bibinfo{volume}{186},
  \bibinfo{pages}{100--139}.
\newblock \DOIprefix\doi{10.1016/j.fss.2011.07.003}.
%Type = Article
\bibitem[{{\'C}iri{\'c} et~al.(2012b){\'C}iri{\'c}, Ignjatovi{\'c},
  Jan\u{c}i{\'c} and Damljanovi{\'c}}]{CiricIJD12}
\bibinfo{author}{{\'C}iri{\'c}, M.}, \bibinfo{author}{Ignjatovi{\'c}, J.},
  \bibinfo{author}{Jan\u{c}i{\'c}, I.}, \bibinfo{author}{Damljanovi{\'c}, N.},
  \bibinfo{year}{2012}b.
\newblock \bibinfo{title}{Computation of the greatest simulations and
  bisimulations between fuzzy automata}.
\newblock \bibinfo{journal}{Fuzzy Sets and Systems} \bibinfo{volume}{208},
  \bibinfo{pages}{22--42}.
\newblock \DOIprefix\doi{10.1016/j.fss.2012.05.006}.
%Type = Article
\bibitem[{\'{C}iri\'{c} et~al.(2010)\'{C}iri\'{c}, Stamenkovi\'{c},
  Ignjatovi\'{c} and Petkovi\'{c}}]{jcss/CIRIC2010}
\bibinfo{author}{\'{C}iri\'{c}, M.}, \bibinfo{author}{Stamenkovi\'{c}, A.},
  \bibinfo{author}{Ignjatovi\'{c}, J.}, \bibinfo{author}{Petkovi\'{c}, T.},
  \bibinfo{year}{2010}.
\newblock \bibinfo{title}{Fuzzy relation equations and reduction of fuzzy
  automata}.
\newblock \bibinfo{journal}{Journal of Computer and System Sciences}
  \bibinfo{volume}{76}, \bibinfo{pages}{609--633}.
\newblock \DOIprefix\doi{10.1016/j.jcss.2009.10.015}.
%Type = Article
\bibitem[{Diaconescu(2020)}]{fuin/Diaconescu20}
\bibinfo{author}{Diaconescu, D.}, \bibinfo{year}{2020}.
\newblock \bibinfo{title}{Modal equivalence and bisimilarity in many-valued
  modal logics with many-valued accessibility relations}.
\newblock \bibinfo{journal}{Fundam. Informaticae} \bibinfo{volume}{173},
  \bibinfo{pages}{177--189}.
\newblock \DOIprefix\doi{10.3233/FI-2020-1920}.
%Type = Article
\bibitem[{Divroodi et~al.(2018)Divroodi, Ha, Nguyen and Nguyen}]{DivroodiHNN18}
\bibinfo{author}{Divroodi, A.}, \bibinfo{author}{Ha, Q.T.},
  \bibinfo{author}{Nguyen, L.}, \bibinfo{author}{Nguyen, H.},
  \bibinfo{year}{2018}.
\newblock \bibinfo{title}{On the possibility of correct concept learning in
  description logics}.
\newblock \bibinfo{journal}{Vietnam J. Computer Science} \bibinfo{volume}{5},
  \bibinfo{pages}{3--14}.
%Type = Article
\bibitem[{Du and Zhu(2018)}]{DuZ18a}
\bibinfo{author}{Du, Y.}, \bibinfo{author}{Zhu, P.}, \bibinfo{year}{2018}.
\newblock \bibinfo{title}{Fuzzy approximations of fuzzy relational structures}.
\newblock \bibinfo{journal}{Int. J. Approx. Reason.} \bibinfo{volume}{98},
  \bibinfo{pages}{1--10}.
\newblock \DOIprefix\doi{10.1016/j.ijar.2018.04.003}.
%Type = Article
\bibitem[{Eleftheriou et~al.(2012)Eleftheriou, Koutras and
  Nomikos}]{EleftheriouKN12}
\bibinfo{author}{Eleftheriou, P.}, \bibinfo{author}{Koutras, C.},
  \bibinfo{author}{Nomikos, C.}, \bibinfo{year}{2012}.
\newblock \bibinfo{title}{Notions of bisimulation for {Heyting}-valued modal
  languages}.
\newblock \bibinfo{journal}{J. Log. Comput.} \bibinfo{volume}{22},
  \bibinfo{pages}{213--235}.
\newblock \DOIprefix\doi{10.1093/logcom/exq005}.
%Type = Article
\bibitem[{Fan and Liau(2014)}]{ai/FanL14}
\bibinfo{author}{Fan, T.}, \bibinfo{author}{Liau, C.}, \bibinfo{year}{2014}.
\newblock \bibinfo{title}{Logical characterizations of regular equivalence in
  weighted social networks}.
\newblock \bibinfo{journal}{Artif. Intell.} \bibinfo{volume}{214},
  \bibinfo{pages}{66--88}.
\newblock \DOIprefix\doi{10.1016/j.artint.2014.05.007}.
%Type = Article
\bibitem[{Fan(2015)}]{Fan15}
\bibinfo{author}{Fan, T.F.}, \bibinfo{year}{2015}.
\newblock \bibinfo{title}{Fuzzy bisimulation for {G\"{o}del} modal logic}.
\newblock \bibinfo{journal}{{IEEE} Trans. Fuzzy Systems} \bibinfo{volume}{23},
  \bibinfo{pages}{2387--2396}.
\newblock \DOIprefix\doi{10.1109/TFUZZ.2015.2426724}.
%Type = Book
\bibitem[{H{\'a}jek(1998)}]{Hajek1998}
\bibinfo{author}{H{\'a}jek, P.}, \bibinfo{year}{1998}.
\newblock \bibinfo{title}{Metamathematics of Fuzzy Logics}.
\newblock \bibinfo{publisher}{Kluwer Academic Publishers}.
%Type = Inproceedings
\bibitem[{Ignjatovi{\'c} et~al.(2015)Ignjatovi{\'c}, {\'C}iri{\'c} and
  Stankovi{\'c}}]{IgnjatovicCS15}
\bibinfo{author}{Ignjatovi{\'c}, J.}, \bibinfo{author}{{\'C}iri{\'c}, M.},
  \bibinfo{author}{Stankovi{\'c}, I.}, \bibinfo{year}{2015}.
\newblock \bibinfo{title}{Bisimulations in fuzzy social network analysis}, in:
  \bibinfo{booktitle}{Proceedings of IFSA-EUSFLAT-15},
  \bibinfo{publisher}{Atlantis Press}. pp. \bibinfo{pages}{404--411}.
%Type = Article
\bibitem[{Jancic(2014)}]{Jancic14}
\bibinfo{author}{Jancic, I.}, \bibinfo{year}{2014}.
\newblock \bibinfo{title}{Weak bisimulations for fuzzy automata}.
\newblock \bibinfo{journal}{Fuzzy Sets Syst.} \bibinfo{volume}{249},
  \bibinfo{pages}{49--72}.
\newblock \DOIprefix\doi{10.1016/j.fss.2013.10.006}.
%Type = Article
\bibitem[{Marti and Metcalfe(2018)}]{aml/MartiM18}
\bibinfo{author}{Marti, M.}, \bibinfo{author}{Metcalfe, G.},
  \bibinfo{year}{2018}.
\newblock \bibinfo{title}{Expressivity in chain-based modal logics}.
\newblock \bibinfo{journal}{Arch. Math. Log.} \bibinfo{volume}{57},
  \bibinfo{pages}{361--380}.
\newblock \DOIprefix\doi{10.1007/s00153-017-0573-4}.
%Type = Article
\bibitem[{Mici{\'c} et~al.(2018)Mici{\'c}, Jan\v{c}i{\'c} and
  Stanimirovi{\'c}}]{MicicJS18}
\bibinfo{author}{Mici{\'c}, I.}, \bibinfo{author}{Jan\v{c}i{\'c}, Z.},
  \bibinfo{author}{Stanimirovi{\'c}, S.}, \bibinfo{year}{2018}.
\newblock \bibinfo{title}{Computation of the greatest right and left invariant
  fuzzy quasi-orders and fuzzy equivalences}.
\newblock \bibinfo{journal}{Fuzzy Sets and Systems} \bibinfo{volume}{339},
  \bibinfo{pages}{99--118}.
\newblock \DOIprefix\doi{10.1016/j.fss.2017.09.004}.
%Type = Article
\bibitem[{Micić et~al.(2022)Micić, Nguyen and Stanimirović}]{MNS.22}
\bibinfo{author}{Micić, I.}, \bibinfo{author}{Nguyen, L.A.},
  \bibinfo{author}{Stanimirović, S.}, \bibinfo{year}{2022}.
\newblock \bibinfo{title}{Characterization and computation of approximate
  bisimulations for fuzzy automata}.
\newblock \bibinfo{journal}{Fuzzy Sets and Systems} \bibinfo{volume}{442},
  \bibinfo{pages}{331--350}.
%Type = Article
\bibitem[{Nguyen(2019)}]{Nguyen-TFS2019}
\bibinfo{author}{Nguyen, L.}, \bibinfo{year}{2019}.
\newblock \bibinfo{title}{Bisimilarity in fuzzy description logics under the
  {Zadeh} semantics}.
\newblock \bibinfo{journal}{{IEEE} Trans. Fuzzy Systems} \bibinfo{volume}{27},
  \bibinfo{pages}{1151--1161}.
\newblock \DOIprefix\doi{10.1109/TFUZZ.2018.2871004}.
%Type = Article
\bibitem[{Nguyen(2021)}]{ijar/Nguyen21}
\bibinfo{author}{Nguyen, L.}, \bibinfo{year}{2021}.
\newblock \bibinfo{title}{Characterizing fuzzy simulations for fuzzy labeled
  transition systems in fuzzy propositional dynamic logic}.
\newblock \bibinfo{journal}{Int. J. Approx. Reason.} \bibinfo{volume}{135},
  \bibinfo{pages}{21--37}.
\newblock \DOIprefix\doi{10.1016/j.ijar.2021.04.006}.
%Type = Article
\bibitem[{Nguyen(2022a)}]{DBLP:journals/jifs/Nguyen22}
\bibinfo{author}{Nguyen, L.}, \bibinfo{year}{2022}a.
\newblock \bibinfo{title}{Computing crisp simulations for fuzzy labeled
  transition systems}.
\newblock \bibinfo{journal}{J. Intell. Fuzzy Syst.} \bibinfo{volume}{42},
  \bibinfo{pages}{3067--3078}.
\newblock \DOIprefix\doi{10.3233/JIFS-210792}.
%Type = Article
\bibitem[{Nguyen(2022b)}]{NguyenFSS2021}
\bibinfo{author}{Nguyen, L.}, \bibinfo{year}{2022}b.
\newblock \bibinfo{title}{Logical characterizations of fuzzy bisimulations in
  fuzzy modal logics over residuated lattices}.
\newblock \bibinfo{journal}{Fuzzy Sets and Systems} \bibinfo{volume}{431},
  \bibinfo{pages}{70--93}.
\newblock \DOIprefix\doi{10.1016/j.fss.2021.08.009}.
%Type = Article
\bibitem[{Nguyen(2023a)}]{isci/Nguyen23}
\bibinfo{author}{Nguyen, L.}, \bibinfo{year}{2023}a.
\newblock \bibinfo{title}{Computing the fuzzy partition corresponding to the
  greatest fuzzy auto-bisimulation of a fuzzy graph-based structure under the
  {G{\"{o}}del} semantics}.
\newblock \bibinfo{journal}{Inf. Sci.} \bibinfo{volume}{630},
  \bibinfo{pages}{482--506}.
\newblock \DOIprefix\doi{10.1016/j.ins.2023.02.029}.
%Type = Article
\bibitem[{Nguyen(2023b)}]{FB4FA}
\bibinfo{author}{Nguyen, L.}, \bibinfo{year}{2023}b.
\newblock \bibinfo{title}{Fuzzy simulations and bisimulations between fuzzy
  automata}.
\newblock \bibinfo{journal}{International Journal of Approximate Reasoning}
  \bibinfo{volume}{155}, \bibinfo{pages}{113--131}.
\newblock \DOIprefix\doi{10.1016/j.ijar.2023.02.002}.
%Type = Misc
\bibitem[{Nguyen(2023c)}]{BFBA-prog}
\bibinfo{author}{Nguyen, L.}, \bibinfo{year}{2023}c.
\newblock \bibinfo{title}{An implementation in {Python} of the agorithms
  provided in the current paper}.
\newblock \bibinfo{howpublished}{Available at
  \url{www.mimuw.edu.pl/~nguyen/BFBA}}.
%Type = Article
\bibitem[{Nguyen et~al.(2020)Nguyen, Ha, Nguyen, Nguyen and Tran}]{FSS2020}
\bibinfo{author}{Nguyen, L.}, \bibinfo{author}{Ha, Q.T.},
  \bibinfo{author}{Nguyen, N.}, \bibinfo{author}{Nguyen, T.},
  \bibinfo{author}{Tran, T.L.}, \bibinfo{year}{2020}.
\newblock \bibinfo{title}{Bisimulation and bisimilarity for fuzzy description
  logics under the {G\"odel} semantics}.
\newblock \bibinfo{journal}{Fuzzy Sets and Systems} \bibinfo{volume}{388},
  \bibinfo{pages}{146--178}.
\newblock \DOIprefix\doi{10.1016/j.fss.2019.08.004}.
%Type = Article
\bibitem[{Nguyen et~al.(2023)Nguyen, Mici{\'c} and
  Stanimirovi{\'c}}]{FuzzyMinimaxNets}
\bibinfo{author}{Nguyen, L.}, \bibinfo{author}{Mici{\'c}, I.},
  \bibinfo{author}{Stanimirovi{\'c}, S.}, \bibinfo{year}{2023}.
\newblock \bibinfo{title}{Fuzzy minimax nets}.
\newblock \bibinfo{journal}{IEEE Transactions on Fuzzy Systems}
  \DOIprefix\doi{10.1109/TFUZZ.2023.3237936}.
%Type = Inproceedings
\bibitem[{Nguyen and Nguyen(2021)}]{fuzzIEEE/NguyenN21}
\bibinfo{author}{Nguyen, L.}, \bibinfo{author}{Nguyen, N.},
  \bibinfo{year}{2021}.
\newblock \bibinfo{title}{Characterizing crisp simulations and crisp directed
  simulations between fuzzy labeled transition systems by using fuzzy modal
  logics}, in: \bibinfo{booktitle}{Proceedings of {FUZZ-IEEE} 2021},
  \bibinfo{publisher}{{IEEE}}. pp. \bibinfo{pages}{1--7}.
\newblock \DOIprefix\doi{10.1109/FUZZ45933.2021.9494504}.
%Type = Article
\bibitem[{Nguyen and Nguyen(2023)}]{tfs/NguyenN23}
\bibinfo{author}{Nguyen, L.}, \bibinfo{author}{Nguyen, N.},
  \bibinfo{year}{2023}.
\newblock \bibinfo{title}{Logical characterizations of crisp bisimulations in
  fuzzy description logics}.
\newblock \bibinfo{journal}{{IEEE} Trans. Fuzzy Syst.} \bibinfo{volume}{31},
  \bibinfo{pages}{1294--1304}.
\newblock \DOIprefix\doi{10.1109/TFUZZ.2022.3198853}.
%Type = Article
\bibitem[{{Nguyen} and {Tran}(2021)}]{TFS2020}
\bibinfo{author}{{Nguyen}, L.}, \bibinfo{author}{{Tran}, D.},
  \bibinfo{year}{2021}.
\newblock \bibinfo{title}{Computing fuzzy bisimulations for fuzzy structures
  under the {G\"odel} semantics}.
\newblock \bibinfo{journal}{IEEE Transactions on Fuzzy Systems}
  \bibinfo{volume}{29}, \bibinfo{pages}{1715--1724}.
\newblock \DOIprefix\doi{10.1109/TFUZZ.2020.2985000}.
%Type = Article
\bibitem[{Nguyen and Tran(2023)}]{CompCB-arxiv}
\bibinfo{author}{Nguyen, L.}, \bibinfo{author}{Tran, D.}, \bibinfo{year}{2023}.
\newblock \bibinfo{title}{Computing crisp bisimulations for fuzzy structures}.
\newblock \bibinfo{journal}{CoRR} \bibinfo{volume}{abs/2010.15671}.
%Type = Article
\bibitem[{Pan et~al.(2014)Pan, Cao, Zhang and
  Chen}]{DBLP:journals/ijar/PanC0C14}
\bibinfo{author}{Pan, H.}, \bibinfo{author}{Cao, Y.}, \bibinfo{author}{Zhang,
  M.}, \bibinfo{author}{Chen, Y.}, \bibinfo{year}{2014}.
\newblock \bibinfo{title}{Simulation for lattice-valued doubly labeled
  transition systems}.
\newblock \bibinfo{journal}{Int. J. Approx. Reason.} \bibinfo{volume}{55},
  \bibinfo{pages}{797--811}.
\newblock \DOIprefix\doi{10.1016/j.ijar.2013.11.009}.
%Type = Article
\bibitem[{Pan et~al.(2015)Pan, Li and Cao}]{DBLP:journals/ijar/PanLC15}
\bibinfo{author}{Pan, H.}, \bibinfo{author}{Li, Y.}, \bibinfo{author}{Cao, Y.},
  \bibinfo{year}{2015}.
\newblock \bibinfo{title}{Lattice-valued simulations for quantitative
  transition systems}.
\newblock \bibinfo{journal}{Int. J. Approx. Reason.} \bibinfo{volume}{56},
  \bibinfo{pages}{28--42}.
\newblock \DOIprefix\doi{10.1016/j.ijar.2014.10.001}.
%Type = Article
\bibitem[{Qiao and Zhu(2021)}]{QiaoZ21}
\bibinfo{author}{Qiao, S.}, \bibinfo{author}{Zhu, P.}, \bibinfo{year}{2021}.
\newblock \bibinfo{title}{Limited approximate bisimulations and the
  corresponding rough approximations}.
\newblock \bibinfo{journal}{Int. J. Approx. Reason.} \bibinfo{volume}{130},
  \bibinfo{pages}{50--82}.
\newblock \DOIprefix\doi{10.1016/j.ijar.2020.12.005}.
%Type = Article
\bibitem[{Qiao et~al.(2023a)Qiao, Zhu and Feng}]{tfs/QIAO2023}
\bibinfo{author}{Qiao, S.}, \bibinfo{author}{Zhu, P.}, \bibinfo{author}{Feng,
  J.e.}, \bibinfo{year}{2023}a.
\newblock \bibinfo{title}{Fuzzy bisimulations for nondeterministic fuzzy
  transition systems}.
\newblock \bibinfo{journal}{IEEE Transactions on Fuzzy Systems}
  \bibinfo{volume}{31}, \bibinfo{pages}{2450--2463}.
\newblock \DOIprefix\doi{10.1109/TFUZZ.2022.3227400}.
%Type = Article
\bibitem[{Qiao et~al.(2023b)Qiao, Zhu and Pedrycz}]{fss/QIAO2023}
\bibinfo{author}{Qiao, S.}, \bibinfo{author}{Zhu, P.},
  \bibinfo{author}{Pedrycz, W.}, \bibinfo{year}{2023}b.
\newblock \bibinfo{title}{Approximate bisimulations for fuzzy-transition
  systems}.
\newblock \bibinfo{journal}{Fuzzy Sets and Systems} ,
  \bibinfo{pages}{108533}\DOIprefix\doi{10.1016/j.fss.2023.108533}.
%Type = Article
\bibitem[{Stanimirovi{\'c} et~al.(2022)Stanimirovi{\'c}, Mici{\'c} and
  {\'C}iri{\'c}}]{SMC.20}
\bibinfo{author}{Stanimirovi{\'c}, S.}, \bibinfo{author}{Mici{\'c}, I.},
  \bibinfo{author}{{\'C}iri{\'c}, M.}, \bibinfo{year}{2022}.
\newblock \bibinfo{title}{Approximate bisimulations for fuzzy automata over
  complete {Heyting} algebras}.
\newblock \bibinfo{journal}{{IEEE} Trans. Fuzzy Syst.} \bibinfo{volume}{30},
  \bibinfo{pages}{437--447}.
\newblock \DOIprefix\doi{10.1109/TFUZZ.2020.3039968}.
%Type = Article
\bibitem[{Stanimirovi{\'c} et~al.(2019)Stanimirovi{\'c}, Stamenkovi{\'c} and
  {\'C}iri{\'c}}]{StanimirovicSC2019}
\bibinfo{author}{Stanimirovi{\'c}, S.}, \bibinfo{author}{Stamenkovi{\'c}, A.},
  \bibinfo{author}{{\'C}iri{\'c}, M.}, \bibinfo{year}{2019}.
\newblock \bibinfo{title}{Improved algorithms for computing the greatest right
  and left invariant boolean matrices and their application}.
\newblock \bibinfo{journal}{Filomat} \bibinfo{volume}{33},
  \bibinfo{pages}{2809--2831}.
%Type = Article
\bibitem[{Wu et~al.(2018a)Wu, Chen, Han and Chen}]{DBLP:journals/ijar/WuCHC18}
\bibinfo{author}{Wu, H.}, \bibinfo{author}{Chen, T.}, \bibinfo{author}{Han,
  T.}, \bibinfo{author}{Chen, Y.}, \bibinfo{year}{2018}a.
\newblock \bibinfo{title}{Bisimulations for fuzzy transition systems
  revisited}.
\newblock \bibinfo{journal}{Int. J. Approx. Reason.} \bibinfo{volume}{99},
  \bibinfo{pages}{1--11}.
\newblock \DOIprefix\doi{10.1016/j.ijar.2018.04.010}.
%Type = Article
\bibitem[{Wu et~al.(2018b)Wu, Chen, Bu and Deng}]{DBLP:journals/fss/WuCBD18}
\bibinfo{author}{Wu, H.}, \bibinfo{author}{Chen, Y.}, \bibinfo{author}{Bu, T.},
  \bibinfo{author}{Deng, Y.}, \bibinfo{year}{2018}b.
\newblock \bibinfo{title}{Algorithmic and logical characterizations of
  bisimulations for non-deterministic fuzzy transition systems}.
\newblock \bibinfo{journal}{Fuzzy Sets Syst.} \bibinfo{volume}{333},
  \bibinfo{pages}{106--123}.
\newblock \DOIprefix\doi{10.1016/j.fss.2017.02.008}.
%Type = Article
\bibitem[{Wu and Deng(2016)}]{DBLP:journals/fss/WuD16}
\bibinfo{author}{Wu, H.}, \bibinfo{author}{Deng, Y.}, \bibinfo{year}{2016}.
\newblock \bibinfo{title}{Logical characterizations of simulation and
  bisimulation for fuzzy transition systems}.
\newblock \bibinfo{journal}{Fuzzy Sets Syst.} \bibinfo{volume}{301},
  \bibinfo{pages}{19--36}.
\newblock \DOIprefix\doi{10.1016/j.fss.2015.09.012}.
%Type = Article
\bibitem[{Yang and Li(2018a)}]{YL.18b}
\bibinfo{author}{Yang, C.}, \bibinfo{author}{Li, Y.}, \bibinfo{year}{2018}a.
\newblock \bibinfo{title}{Approximate bisimulation relations for fuzzy
  automata}.
\newblock \bibinfo{journal}{Soft Computing} \bibinfo{volume}{22},
  \bibinfo{pages}{4535--4547}.
%Type = Article
\bibitem[{Yang and Li(2018b)}]{YL.18a}
\bibinfo{author}{Yang, C.}, \bibinfo{author}{Li, Y.}, \bibinfo{year}{2018}b.
\newblock \bibinfo{title}{$\epsilon$-bisimulation relations for fuzzy
  automata}.
\newblock \bibinfo{journal}{IEEE Transactions on Fuzzy Systems}
  \bibinfo{volume}{26}, \bibinfo{pages}{2017--2029}.
%Type = Article
\bibitem[{Yang and Li(2020)}]{DBLP:journals/fss/YangL20}
\bibinfo{author}{Yang, C.}, \bibinfo{author}{Li, Y.}, \bibinfo{year}{2020}.
\newblock \bibinfo{title}{Approximate bisimulations and state reduction of
  fuzzy automata under fuzzy similarity measures}.
\newblock \bibinfo{journal}{Fuzzy Sets Syst.} \bibinfo{volume}{391},
  \bibinfo{pages}{72--95}.
\newblock \DOIprefix\doi{10.1016/j.fss.2019.07.010}.

\end{thebibliography}
\bibliographystyle{elsarticle-harv}

%===============================================================================

\newpage
\appendix

\section{Additional Proofs}

The proofs given in this appendix are similar to the ones of~\cite{FB4FA} (as depth-bounded fuzzy simulations/bisimulations are approximations of fuzzy simulations/bisimulations). We present them here to make the article self-contained. 

\begin{proof}[of Lemma~\ref{lemma: KHSBS}]
	Let $\Phi = (\varphi_n)_{n \in \NN}$ be a depth-bounded fuzzy simulation between~$\mA$ and~$\mAp$.
	We prove the lemma by induction on the structure of~$\alpha$. 
	\begin{itemize}
		\item Case $\alpha = \tau$: The assertion follows from~\eqref{eq: HFKJA 1}, \eqref{fop: GDJSK 10} and \eqref{eq: HFKJA 2}.
		
		\item Case $\alpha = (s \circ \beta)$: We have $1 \leq |\alpha| \leq n$. By~\eqref{eq: HFKJA 3}, \eqref{fop: GDJSK 10} and the induction assumption, we have 
		\[ \varphi_n^{-1} \circ \alpha^\mA = \varphi_n^{-1} \circ \deltaA_s \circ \beta^\mA \leq \deltaAp_s \circ \varphi_{n-1}^{-1} \circ \beta^\mA \leq  \deltaAp_s \circ \beta^\mAp = \alpha^\mAp. \]
		
		\item Case $\alpha = (a \to \beta)$: By~\eqref{fop: GDJSK 80a}, \eqref{fop: GDJSK 20} and the induction assumption, we have  
		\[ \varphi_n^{-1} \circ \alpha^\mA = \varphi_n^{-1} \circ (a \fto \beta^\mA)  \leq (a \fto \varphi_n^{-1} \circ \beta^\mA) \leq (a \fto \beta^\mAp) = \alpha^\mAp. \]
		
		\item Case $\alpha = (\beta \land \gamma)$: By~\eqref{fop: GDJSK 10} and the induction assumption, we have 
		\[ \varphi_n^{-1} \circ \alpha^\mA = \varphi_n^{-1} \circ (\beta^\mA \land \gamma^\mA) \leq (\varphi_n^{-1} \circ \beta^\mA) \land (\varphi_n^{-1} \circ \gamma^\mA) \leq \beta^\mAp \land \gamma^\mAp = \alpha^\mAp. \]
	\end{itemize}

	\vspace{-2.5em}
	
	\myend	
\end{proof}

\begin{proof}[of Theorem~\ref{theorem: HDJHA}]
	By the consequence~\eqref{eq: HDJHS} of Lemma~\ref{lemma: KHSBS}, it suffices to prove that $\Phi = (\varphi_n)_{n \in \NN}$ is a depth-bounded fuzzy simulation between $\mA$ and $\mAp$. 
	Since $\mFdbs{n-1} \subseteq \mFdbs{n}$ for $n \geq 1$, $\Phi$ satisfies the condition~\eqref{eq: HFKJA 1}. 
	By definition, for every $n \in \NN$ and $\tuple{x,x'} \in A \times A'$, 
	\[
	\varphi_n(x,x') \leq (\tau^\mA(x) \fto \tau^\mAp(x')), 
	\]
	which implies 
	\(
	\varphi_0(x,x') \fand \tauA(x) \leq \tauAp(x'). 
	\)
	Therefore, \eqref{eq: HFKJA 2} holds. To prove~\eqref{eq: HFKJA 3}, it suffices to show that, for every $n \geq 1$, $s \in \Sigma$, $\tuple{x',y} \in A' \times A$ and $x \in A$, there exists $y' \in A'$ such that  
	\[ \varphi_n(x,x') \fand \deltaA_s(x,y) \leq \deltaAp_s(x',y') \fand \varphi_{n-1}(y,y'). \]
	For a contradiction, suppose that there exist $n \geq 1$, $s \in \Sigma$, $\tuple{x',y} \in A' \times A$ and $x \in A$ such that, for every $y' \in A'$,   
	\[ \varphi_n(x,x') \fand \deltaA_s(x,y) > \deltaAp_s(x',y') \fand \varphi_{n-1}(y,y'). \] 
	Since $\fand$ is continuous, it follows that, for every $y' \in A'$, there exists $\alpha_{y'} \in \mFdbs{n-1}$ such that  
	\[ \varphi_n(x,x') \fand \deltaA_s(x,y) > \deltaAp_s(x',y') \fand (\alpha_{y'}^\mA(y) \fto \alpha_{y'}^\mAp(y')). \] 
	As $\mAp$ is image-finite, let $y'_1,\ldots,y'_m$ be all elements of $A'$ such that $\deltaAp_s(x',y') > 0$. For $1 \leq i \leq m$, let $\beta_{y'_i} = (\alpha_{y'_i}^\mA(y) \to \alpha_{y'_i})$. We have that, for every $1 \leq i \leq m$, $\beta_{y'_i}^\mA(y) = 1$ (by~\eqref{fop: GDJSK 30}) and 
	\[ \varphi_n(x,x') \fand \deltaA_s(x,y) > \deltaAp_s(x',y'_i) \fand \beta_{y'_i}^\mAp(y'_i). \] 
	Since $\mL$ is linear, it follows that 
	\begin{equation}\label{eq: HDJAA}
		\varphi_n(x,x') \fand \deltaA_s(x,y) > \bigvee_{1 \leq i \leq m}\!(\deltaAp_s(x',y'_i) \fand \beta_{y'_i}^\mAp(y'_i)).
	\end{equation}
	Let $\alpha = s \circ (\beta_{y'_1} \land \ldots \land \beta_{y'_m})$. 
	We have $\alpha \in \mFdbs{n}$. 
	By~\eqref{fop: GDJSK 10} and~\eqref{fop: GDJSK 40}, we have 
	\begin{eqnarray*}
		\alpha^\mA(x) & \geq & \deltaA_s(x,y) \\
		\alpha^\mAp(x') & \leq & \bigvee_{1 \leq i \leq m}\!(\deltaAp_s(x',y'_i) \fand \beta_{y'_i}^\mAp(y'_i)).
	\end{eqnarray*}
	By~\eqref{eq: HDJAA} and~\eqref{fop: GDJSK 10}, it follows that 
	\( \varphi_n(x,x') \fand \alpha^\mA(x) > \alpha^\mAp(x'), \)
	which is equivalent to 
	\( \varphi_n(x,x') > (\alpha^\mA(x) \fto \alpha^\mAp(x')). \)
	This contradicts the definition of~$\varphi_n$. 
	\myend
\end{proof}

\begin{proof}[of Lemma~\ref{lemma: KHSBS 2}]
	Let $\Phi = (\varphi_n)_{n \in \NN}$ be a depth-bounded fuzzy bisimulation between~$\mA$ and~$\mAp$. 
	We prove the lemma by induction on the structure of~$\alpha$. 
	The cases where $\alpha$ is of the form $\tau$, $(s \circ \beta)$ or $(\beta \land \gamma)$ can be dealt with analogously as done in the proof of Lemma~\ref{lemma: KHSBS}. Consider the case $\alpha = (a \leftrightarrow \beta)$. Let $\tuple{x,x'}$ be an arbitrary pair from $A \times A'$. 
	By the induction assumption (i.e., \eqref{eq: KHSBS 2a} and~\eqref{eq: KHSBS 2b} with $\alpha$ replaced by $\beta$), 
	\[ \varphi_n(x,x') \leq (\beta^\mA(x) \fequiv \beta^\mAp(x')). \]
	By~\eqref{fop: GDJSK 150b}, it follows that 
	\[ \varphi_n(x,x') \leq ((a \fequiv \beta^\mA(x)) \fequiv (a \fequiv \beta^\mAp(x'))), \]
	which means
	\[ \varphi_n(x,x') \leq (\alpha^\mA(x) \fequiv \alpha^\mAp(x')). \]
	As this holds for all $\tuple{x,x'} \in A \times A'$, we can derive \eqref{eq: KHSBS 2a} and~\eqref{eq: KHSBS 2b}.
	\myend
\end{proof}

\begin{proof}[of Theorem~\ref{theorem: HDJHA 2}]
	This proof is similar to the proof of Theorem~\ref{theorem: HDJHA}. By the consequence~\eqref{eq: HDJHS 2} of Lemma~\ref{lemma: KHSBS 2}, it suffices to prove that $\Phi = (\varphi_n)_{n \in \NN}$ is a depth-bounded fuzzy bisimulation between $\mA$ and $\mAp$. 
	
	Since $\mFdbbs{n-1} \subseteq \mFdbbs{n}$ for $n \geq 1$, $\Phi$ satisfies the condition~\eqref{eq: HFKJA 1}. 
	
	By definition, for every $n \in \NN$ and $\tuple{x,x'} \in A \times A'$, 
	\[
	\varphi_n(x,x') \leq (\tau^\mA(x) \fequiv \tau^\mAp(x')), 
	\]
	which implies 
	\begin{eqnarray*}
		\varphi_0(x,x') \fand \tauA(x) & \leq & \tauAp(x') \\
		\varphi_0(x,x') \fand \tauAp(x') & \leq & \tauA(x).
	\end{eqnarray*}
	Therefore, \eqref{eq: HFKJA 2} and~\eqref{eq: HFKJA 4} hold. 
	
	To prove~\eqref{eq: HFKJA 3}, it suffices to show that, for every $n \geq 1$, $s \in \Sigma$, $\tuple{x',y} \in A' \times A$ and $x \in A$, there exists $y' \in A'$ such that  
	\[ \varphi_n(x,x') \fand \deltaA_s(x,y) \leq \deltaAp_s(x',y') \fand \varphi_{n-1}(y,y'). \]
	For a contradiction, suppose that there exist $n \geq 1$, $s \in \Sigma$, $\tuple{x',y} \in A' \times A$ and $x \in A$ such that, for every $y' \in A'$,   
	\[ \varphi_n(x,x') \fand \deltaA_s(x,y) > \deltaAp_s(x',y') \fand \varphi_{n-1}(y,y'). \] 
	Since $\fand$ is continuous, it follows that, for every $y' \in A'$, there exists $\alpha_{y'} \in \mFdbbs{n-1}$ such that  
	\[ \varphi_n(x,x') \fand \deltaA_s(x,y) > \deltaAp_s(x',y') \fand (\alpha_{y'}^\mA(y) \fequiv \alpha_{y'}^\mAp(y')). \] 
	As $\mAp$ is image-finite, let $y'_1,\ldots,y'_m$ be all elements of $A'$ such that $\deltaAp_s(x',y') > 0$. For $1 \leq i \leq m$, let $\beta_{y'_i} = (\alpha_{y'_i}^\mA(y) \leftrightarrow \alpha_{y'_i})$. We have that, for every $1 \leq i \leq m$, $\beta_{y'_i}^\mA(y) = 1$ (by~\eqref{fop: GDJSK 30}) and 
	\[ \varphi_n(x,x') \fand \deltaA_s(x,y) > \deltaAp_s(x',y'_i) \fand \beta_{y'_i}^\mAp(y'_i). \] 
	Since $\mL$ is linear, it follows that 
	\begin{equation}\label{eq: HDJAA 2}
		\varphi_n(x,x') \fand \deltaA_s(x,y) > \bigvee_{1 \leq i \leq m}\!(\deltaAp_s(x',y'_i) \fand \beta_{y'_i}^\mAp(y'_i)).
	\end{equation}
	Let $\alpha = s \circ (\beta_{y'_1} \land \ldots \land \beta_{y'_m})$. 
	We have $\alpha \in \mFdbbs{n}$. 
	By~\eqref{fop: GDJSK 10} and~\eqref{fop: GDJSK 40}, we have 
	\begin{eqnarray*}
		\alpha^\mA(x) & \geq & \deltaA_s(x,y) \\
		\alpha^\mAp(x') & \leq & \bigvee_{1 \leq i \leq m}\!(\deltaAp_s(x',y'_i) \fand \beta_{y'_i}^\mAp(y'_i)).
	\end{eqnarray*}
	By~\eqref{eq: HDJAA 2} and~\eqref{fop: GDJSK 10}, it follows that 
	\( \varphi_n(x,x') \fand \alpha^\mA(x) > \alpha^\mAp(x'), \)
	which is equivalent to 
	\( \varphi_n(x,x') > (\alpha^\mA(x) \fto \alpha^\mAp(x')). \)
	This contradicts the definition of~$\varphi_n$. 
	
	The assertion~\eqref{eq: HFKJA 5} can be proved analogously. 
	\myend
\end{proof}

%===============================================================================

\end{document}